\documentclass[10pt,a4paper]{article}
\usepackage[latin1]{inputenc}
\usepackage{amsmath}
\usepackage{amsfonts}
\usepackage{amssymb}
\usepackage{amsthm}
\usepackage{graphicx}
\usepackage{enumitem}
\usepackage{epstopdf}

\usepackage{float} 

\usepackage{soul} 

\usepackage[a4paper,margin=2.0cm]{geometry}

\usepackage{caption}
\usepackage{natbib}
\usepackage{aas_macros}

\usepackage{hyperref}

\usepackage{pdflscape}	

\usepackage{siunitx}

\DeclareSIUnit\Msol{M_\odot}
\DeclareSIUnit\pc{pc}
\DeclareSIUnit\kpc{\kilo\pc}
\DeclareSIUnit\yr{yr}
\DeclareSIUnit\Myr{\mega\yr}
\DeclareSIUnit\Gyr{\giga\yr}
\DeclareSIUnit\AU{AU}

\newcommand*\diff{\mathop{}\!\mathrm{d}} 
\newcommand*\R{\mathbb R} 
\newcommand*\N{\mathbb N} 
\newcommand*\measure{\mathcal L}
\newcommand*\scriptQ{\mathcal Q}
\newcommand*\Epot{E_{pot}}
\newcommand*\Epottilde{\tilde E_{pot}}
\newcommand*\Ekin{E_{kin}}
\newcommand*\Etot{E_{tot}}
\newcommand*\He{\mathcal H_E}
\newcommand*\Hb{\mathcal H_B}

\newcommand\VQMS{(VQMS)}

\DeclareMathOperator\supp{supp}
\DeclareMathOperator\divergence{div}

\newtheorem{thm}{Theorem}[section]
\newtheorem{lem}[thm]{Lemma}
\newtheorem{prop}[thm]{Proposition}

\newtheorem*{generalassumptions}{General assumptions}
\newtheorem*{assumptionsPsi}{Assumptions on $\Psi$}
\newtheorem*{assumptionsPhi}{Assumptions on $\Phi$}
\newtheorem*{assumptionsLambda}{Assumptions on $\lambda$}

\theoremstyle{definition}
\newtheorem{defn}[thm]{Definition}
\theoremstyle{remark}
\newtheorem*{rem*}{Remark}
\newtheorem{rem}[thm]{Remark}

\numberwithin{equation}{section}

\title{\LARGE\textbf{Non-Linear Stability of Spherical Models in MOND}}
\author{Joachim Frenkler\\ Fakult\"at f\"ur Mathematik, Physik und Informatik\\Universit\"at Bayreuth \\ D-95440 Bayreuth, Germany \\ joachim.frenkler@uni-bayreuth.de}
\date{\today}

\graphicspath{{./}}
\usepackage{subfiles}

\begin{document}
	
	\maketitle
	
	\begin{abstract}
		We prove the non-linear stability of a large class of spherically symmetric equilibrium solutions of the collisonless Boltzmann equation and of the Euler equations in MOND. This is the first such stability result in MOND that is proven with mathematical rigour. Several difficulties that arise from the intrinsic non-linearity of the Mondian field equations we solve by restricting our analysis to spherical symmetry. We discuss every point where this extra assumption was necessary and outline how in future works one could get along without it. In the end we show at a the example of a polytropic model how our stability result can be applied.
	\end{abstract}
	
	\section{Introduction}
	
	One of the biggest mysteries in astrophysics is the nature of dark matter: Is it an unobserved form of matter? Or is it just a phantom due a the breakdown of Newtonian physics on galactic scales? Some 40 years ago \cite{1983ApJ...270..365Milgrom} developed MOND (MOdified Newtonian Dynamics), a theory that proposes that below a critical acceleration $a_0 \approx \SI{1.2e-10}{\m\per\s\squared}$ Newtonian physics are no longer valid. If the gravitational acceleration $g_N$ of an object expected from Newtonian physics would be much smaller than $a_0$, then MOND predicts that the real gravitational acceleration $g_{real}$ of this object is actually given by $g_{real} = \sqrt{a_0g_N}$.	With its single modification MOND is capable of explaining many astrophysical phenomena usually attributed to dark matter \citep{2012LRR....15...10FamaeyMcGaugh}, most notably flat rotation curves \citep{2011A&A...527A..76GentileFamaeyBlok}. 
	
	In this paper we couple both the collisionless Boltzmann equation\footnote{In mathematics the collisionless Boltzmann equation is often called the Vlasov equation.} and the Euler equations with the Mondian field equations. With mathematical rigour we prove the non-linear stability of a large class of equilibrium solutions of these equations. This is the first time that such results are proven using Mondian physics. In contrast to Newtonian physics, where several such stability results exist \citep{1999GuoRein,2007Rein,2012InMat.187..145LemouMehatsRaphael}, there is only a numerical stability analysis available in MOND for spherical objects \citep{2011Nipoti}. The only analytical, Mondian stability treatments we are aware of are the linearized analyses of the Toomre instability in disc-like systems by \cite{1989ApJ...338..121Milgrom} and \cite{2018arXiv180810545Banik}. Objects that can be described with models similar to the ones treated in this paper cover for example globular cluster \citep{2011ApJ...743...43Ibata,2012MNRAS.422L..21Sanders}, dwarf spheriodal galaxies and galaxy groups \citep{2021Milgrom}.
	

	
	Our final stability theorems are not the only benefit one can dervie from the present paper. Across this paper, we develop new, genuinely Mondian ideas how to answer the arising mathematical questions. For example there is one point in our argumentation where it is necessary to prove that the mass of a certain sequence\footnote{The models we discuss are constructed as minimizers of a variational problem. To find such a minimizer one takes first a minimizing sequence, subsequently shows that this sequence converges and finally proves that the limit is indeed a minimizer. It is important to prove that the mass of such a minimizing sequence always remains concentrated.} remains concentrated. Here, our proof explicitly uses the differences between Mondian and Newtonian physics and this way becomes more direct then corresponding proofs in Newtonian physics (Lemma \ref{lemma masses remain concentrated along minimizing sequences}) Further there are many technical difficulties that we have to solve and which originate from the intrinsic non-linear nature of the Mondian field equations. We solve all these arising difficulties (see, e.g., Lemma \ref{lemma estimates for Epot}, Theorem \ref{thm existence of minimizers of Hr} or Lemma \ref{lemma derivative of EpotQ}). Several of them made it necessary to restrict our results to spherical symmetry. We do not linger on this restricting, extra assumption, but discuss all points where the assumption of spherical symmetry was crucial. For several points we outline how they could be solved without it. For the others we illustrate how hard it might be to find a solution (Section \ref{section getting rid of spherical symmetry}).
	Thus the present paper is both an important step forward toward a better understanding of the stability of spherical systems in MOND and a pool of new ideas how problems arising from Mondian physics can be solved with mathematical rigour.
	
	Before we start with strict mathematics, we state in Sections \ref{principle results and basic ideas collisionless} and \ref{principle results and basic ideas fluid} our principle stability results in an informal way and explain the basic ideas that lead to these results in a simple language. In Section \ref{section Potential theory} we collect and prove all results from Newtonian and Mondian potential theory that are necessary for the subsequent sections. In Section \ref{section finding minimizers} we study two variational problems and prove the existence of minimizers for both problems. Further, we study the regularity of these minimizers. In Section \ref{section stability} we prove two stability results: One for equilibrium solutions of the collisionless Boltzmann equation and one for equilibrium solutions of the Euler equations. In Section \ref{section getting rid of spherical symmetry} we discuss all points where the assumption of spherical symmetry was crucial. In the last section, Section \ref{section minimizer must be unique}, we demonstrate how to apply our results to prove the non-linear stability of a particular model.
	
	\subsection{Principle results and basic ideas - Collisionless models} \label{principle results and basic ideas collisionless}
	
	Instead of following the trajectories of all stars in a globular cluster or a galaxy, we model the evolution of these systems with a distribution function $f=f(t,x,v)$ that evolves on position-velocity space; $t\in[0,\infty)$ is the time variable and $x,v\in\R^3$ are the position and velocity variable. $f$ shall solve the collisionless Boltzmann equation
	\begin{equation} \label{Boltzmann equation}
		\partial_t f + v\cdot \nabla_x f - \nabla U^M_f \cdot \nabla_v f = 0
	\end{equation}
	where $U^M_f$ is the Mondian gravitational potential created by $f$ itself (for the exact definition of $U^M_f$ see the next section). Equation \eqref{Boltzmann equation} is a transport equation and states that $f$ is constant along solutions of the ordinary differential equation
	\begin{align} \label{characteristic system}
		\dot{x} & = v, \\
		\dot{v} & = -\nabla U^M_{f(t)}(x). \nonumber
	\end{align}
	A non-negative, spherically symmetric equilibrium solution $f_0=f_0(x,v)$ of \eqref{Boltzmann equation} is a good model for a globular cluster. Such models were, e.g., used in the vivid discussion whether the globular cluster NGC 2419 does or does not contradict MOND \citep{2011ApJ...743...43Ibata,2012MNRAS.422L..21Sanders}. A spheroidal galaxy can be modelled in the same way, however \cite{2021Milgrom} chose to model such systems with a fluid approach. We prove stability for such fluid models too; see Section \ref{principle results and basic ideas fluid}.
	
	Since our model $f_0$ shall be time independent also the gravitational potential $U^M_{f_0}$ is time independent and thus the local energy
	\begin{equation*}
		E(x,v) = \frac{|v|^2}{2}  + U^M_{f_0}(x), \quad x,v\in\R^3,
	\end{equation*}
	is constant along solutions of \eqref{characteristic system}. Thus it is a good idea to search models that are functions of the local energy. This way we guarantee that the model is constant along solutions  of \eqref{characteristic system} and thus solves \eqref{Boltzmann equation}. To guarantee that the resulting model has finite mass and support it is convenient to introduce a cut-off energy $E_0\in\R$ and search for models of the form
	\begin{equation} \label{ansatz}
		f_0(x,v) =
		\begin{cases}
			\phi(E_0 - E), &  E < E_0, \\
			0, & E \geq E_0
		\end{cases}
	\end{equation}
	where $\phi:[0,\infty)\rightarrow[0,\infty)$ is an ansatz function.
	
	\cite{2015Rein} proved the existence of a large class of such models but with the methods used there it was not possible to analyse further the stability of the models. To be able to do so, we do a trick here. We rewrite the first line of \eqref{ansatz} in the following way:
	\begin{equation*}
		f_0 = \phi(E_0-E) \quad \Leftrightarrow \quad E + \phi^{-1}(f_0) = E_0 \cdot 1.
	\end{equation*}
	While the form on the left hand side is our ansatz to construct an equilibrium solution of \eqref{Boltzmann equation}, the right hand side can be interpreted in a completely different way. The local energy $E$ is the functional derivative of the total energy of the entire system $\Etot(f_0)$. We give a precise definition of the total energy below. $\phi^{-1}(f_0)$ is the functional derivative of the Casimir functional $\mathcal C(f_0) = \iint \Phi(f_0) \diff x\diff v$\footnote{If not stated otherwise, the integral sign $\int$ extends always over the entire space $\R^3$.}, where $\Phi$ is the antiderivative of $\phi^{-1}$. $E_0$ takes the role of a Lagrange multiplier. And the number $1$ is the functional derivative of the mass of the system $\iint f_0\diff x \diff v$. Thus the right hand side can be interpreted as the Euler-Lagrange equation of the variational problem
	\begin{equation} \label{variational problem}
		\text{minimize } \Etot(f) + \mathcal C(f) \text{ s.t. } \iint f \diff x \diff v = M, f\geq 0
	\end{equation}
	for some prescribed mass $M>0$. The additional constraint $f\geq 0$ is necessary because mass must be non-negative. Further it guarantees that a minimizers of the above variational problem vanishes if $E\geq E_0$ as demanded in \eqref{ansatz}.
	
	Now we have a new possibility to construct equilibrium solutions of  \eqref{Boltzmann equation}. Instead of solving \eqref{ansatz} directly (as it was done in \cite{2015Rein}), we can search a minimizer of the variational problem \eqref{variational problem}. This way is more costly in in terms of labour, but it guarantees additionally the non-linear stability of the equilibrium solution. The basic idea behind the stability result is the following:
	
	If we have initial data $f(0)$ that is -- in a certain sense -- close to the minimizer $f_0$ of \eqref{variational problem}, then also $\Etot(f(0))$ and $\mathcal C(f(0))$ are close to $\Etot(f_0)$ and $\mathcal C(f_0)$. Since $\Etot$ and $\mathcal{C}$ are conserved quantities of solutions of \eqref{Boltzmann equation}, $\Etot(f(t))$ and $\mathcal{C}(f(t))$ will stay close to $\Etot(f_0)$ and $\mathcal C(f_0)$ for all times $t>0$. From the two quantities $\Etot$ and $\mathcal C$, we can construct tools to measure distances between the minimizer $f_0$ and disturbed distribution functions $f$. Combing these distance measures with the above conservation laws and the fact that $f_0$ is a minimizer, we can prove that not only $f(0)$ was close to $f_0$ initially, but that $f(t)$ stays close to $f_0$ for all times $t>0$.
	
	To prove such a stability statement we make the following assumptions on $\Phi$:
	
		\begin{assumptionsPhi}
		$\Phi\in C^1([0,\infty))$, $\Phi(0)=\Phi'(0)=0$ and it holds:
		\begin{enumerate}[label=($\Phi$\arabic*)]
			\item $\Phi$ is strictly convex, \label{Phi convex}
			\item $\Phi(f)\geq Cf^{1+1/k}$ for $f\geq 0$ large, where $0 < k < 3/2$. \label{Phi of f is large for f large}
		\end{enumerate}
		
	\end{assumptionsPhi}

	In view of \eqref{ansatz} it is useful to translate the assumptions on $\Psi$ also in assumptions on $\phi$. Using that $\phi = (\Phi')^{-1}$ the above assumptions imply that $\phi$ is continuous and strictly increasing with $\phi(0) = 0$ and $\phi(\eta) \leq C\eta^k$ for $\eta$ large\footnote{To see that $\phi$ is indeed bounded by $C\eta^k$ we need to use both \ref{Phi convex} and \ref{Phi of f is large for f large}. Compare the analog argument for the function $\Psi$ in the proof of Lemma \ref{lemma minimizers are continuous}.}.
 	With the above assumptions the following theorem holds:
	
	\begin{thm} \label{thm summarization vlasov}
		A minimizer $f_0$ of the variational problem \eqref{variational problem} exists if we limit the variational problem to spherically symmetric distribution functions $f$. This minimizer is of the form \eqref{ansatz} where $\phi=(\Phi')^{-1}$. If this minimizer is unique, then every spherically symmetric solution $f(t)$ of \eqref{Boltzmann equation} with initial data close to $f_0$ remains close to $f_0$ for all times $t>0$.
	\end{thm}

	Some remarks to this theorem. An important ingredient in proving Theorem \ref{thm summarization vlasov} is that the Mondian potential $U^M(x)$ diverges when $|x|\rightarrow \infty$. This is different from Newtonian potentials that tend to zero when $|x|\rightarrow\infty$. This way Mondian potentials confine mass more effectively. For this reason we need less assumptions on $\Phi$ than in the Newtonian situation and it even makes the proof of Theorem \ref{thm summarization vlasov} more direct than in comparable Newtonian situations (compare the proof of an analog theorem in the Newtonian situation in \cite{2007Rein}). However the intrinsic non-linearity of the Mondian field equations introduces several technical difficulties. This made it necessary to limit Theorem \ref{thm summarization vlasov} to spherically symmetric perturbations at first. We discuss the reasons for this and which efforts must be undertaken to get rid of this restriction in detail in Section \ref{section getting rid of spherical symmetry}. Further, we limited our stability statement to minimizers that are unique. However, this is not a strong restriction. In Section \ref{section minimizer must be unique}, we apply our stability result to a polytropic model and illustrate at this example how one can deal with the uniqueness assumption.
	
	And a remark about the quality of our stability result. Theorem \ref{thm summarization vlasov} is a full, non-linear stability statement. Such a result is much stronger than a linearized stability result would be. For a linearized stability result one would first simplify the collisionless Boltzmann equation \eqref{Boltzmann equation} by linearizing it. We do nothing the like. We treat the full, non-linear, collisionless Boltzmann equation and prove that solutions of the non-linear equation stay close to the minimizer if they start close to the minimizer. Theorem \ref{thm summarization vlasov} is the first stability result in Mondian physics of this quality (For more details on the quality of different stability results see \cite{2003MNRAS.344.1296ReinGuo}).
	
	\subsection{Principle results and basic ideas - Fluid models} \label{principle results and basic ideas fluid}
	
	\cite{2021Milgrom} used fluid models to describe dwarf spheroidal galaxies. In such a model a galaxy is described with a density $\rho(t,x)$ that evolves on position space, a velocity field $u(t,x)$ and a pressure $p(t,x)$. Comparing this to the above collisionless situation, $\rho$ is related to $f$ via
	\begin{equation*}
		\rho(t,x) = \int f(t,x,v) \diff v,
	\end{equation*}
	$u(t,x)$ is the mean velocity of $f$ at position $x$, i.e.,
	\begin{equation*}
		u(t,x) = \int v f(t,x,v) \diff v,
	\end{equation*}
	and the effect of the velocity dispersion on the evolution of the system is attributed to the pressure $p$, which is related to $\rho$ via an equation of state. $\rho$, $u$ and $p$ shall solve the following Euler equations
	\begin{align} \label{Euler equations}
		& \partial_t \rho + \divergence(\rho u) = 0, \\
		& \rho \,\partial_t u + \rho(u \cdot \nabla) u = - \nabla p - \rho\,\nabla U^M_\rho. \nonumber
	\end{align}
	
	As, e.g., done by \cite{2021Milgrom}, we can construct an equilibrium solution $\rho_0(x)$ with vanishing velocity field $u_0(x) = 0$ of these equations via the ansatz
	\begin{equation} \label{ansatz equilibrium solution fluid models}
		\rho_0(x) =
		\begin{cases}
		\psi(E_0 - U^M_{\rho_0}(x)), & \text{if } U^M_{\rho_0} < E_0, \\
		0, & \text{if } U^M_{\rho_0} \geq E_0,
		\end{cases}
	\end{equation}
	where $E_0\in\R$ is a cut-off energy and $\psi$ is a certain ansatz function. Rewriting this ansatz as above in the collisionless situation, we see that we can construct such an equilibrium solution as the minimizer of the variational problem
	\begin{equation} \label{variational problem fluid}
		\text{minimize } \Epot^M(\rho) + \mathcal C(\rho) \text{ s.t. } \int \rho \diff x = M, \rho \geq 0,
	\end{equation}
	where $M>0$ is some prescribed mass. Here, $\Epot(\rho)$ is the (Mondian) potential energy of $\rho$ due to gravity. Analoug to above, $\mathcal C(\rho) = \int \Psi(\rho(x)) \diff x$ with $\Psi$ beinig the antiderivative of $(\psi')^{-1}$. Thus $\mathcal C$ has the form of a Casimir functional, but now in the fluid situation it represents the energy attributed to the pressure. $\int \rho\diff x$ is the mass of the model. Observe, that in contrast to above we do not need to minimize the kinetic part of the energy of the system because it is obiviously minimal when the velocity field vanishes (above in the collisionless situation the kinetic energy is hidden in the total energy $\Etot$). The equation of state for the so constructed density $\rho_0$ is
	\begin{equation} \label{equ equation of state in introduction}
		p(x) = \rho \Psi'(\rho) - \Psi(\rho).
	\end{equation}
	Under the following assumptions on $\Psi$, we prove a non-linear stability result for equilibrium solutions that are of the form \eqref{ansatz equilibrium solution fluid models}
	
	\begin{assumptionsPsi}
		$\Psi \in C^1([0,\infty))$, $\Psi(0)=\Psi'(0)=0$ and it holds:
		\begin{enumerate}[label=($\Psi$\arabic*)]
			\item $\Psi$ is strictly convex, \label{Psi convex}
			\item $\Psi(\rho)\geq C\rho^{1+1/n}$ for $\rho> 0$ large, where $0 < n < 3$. \label{Psi of rho is large for rho large}
		\end{enumerate}
	\end{assumptionsPsi}
	
	With these assumptions, the following theorem holds
	
	\begin{thm} \label{thm summarization fluid}
		A minimizer $\rho_0$ of the variational problem \eqref{variational problem fluid} exists if we limit the variational problem to spherically symmetric densities $\rho$. This minimizer is of the form \eqref{ansatz equilibrium solution fluid models} where $\psi=(\Psi')^{-1}$. If this minimizer is unique, then every spherically symmetric solution $(\rho(t),u(t))$ of \eqref{Euler equations} that conserves energy and that has initial data close to $(\rho_0,0)$ remains close to $(\rho_0,0)$ for all times $t>0$.
	\end{thm}

	The same remarks as for Theorem \ref{thm summarization vlasov} apply.
	
	\section{Potential theory} \label{section Potential theory}
	
	In the previous two sections we have already used the notions $U^M$ and $\Epot^M$ for the gravitational potential and the potential energy in MOND. In this section we introduce these notions precisely.
	
	First we have to take a look on Newtonian potentials. For a given density $\rho(x)$, the corresponding Newtonian potential $U^N_{\rho}(x)$ is the solution of
	\begin{equation} \label{Poisson equation}
		\Delta U^N_\rho = 4\pi G \rho, \quad \lim_{|x|\rightarrow\infty} U^N_\rho(x) = 0;
	\end{equation}
	since the exact value of the gravitational constant $G$ does not affect our analysis, we set $G=1$. According to Newtonian physics we would expect that an object, which is located at position $x\in\R^3$, feels an acceleration $g_N= -\nabla U^N_\rho(x)$.
	The basic MOND paradigm states that this is only true for large accelerations. Depending on whether $g_N$ is lower or larger than $a_0$, in Mondian physics $g_N$ and the real acceleration $g_{real}$ are related via
	\begin{equation*}
		\begin{array}{ll}
		g_{real} \approx \sqrt{a_0g_N} & \text{if } g_N \ll a_0, \\
		g_{real} \approx g_N, & \text{if } g_N \gg a_0.
		\end{array}
	\end{equation*}
	We have to precise how $g_{real}$ and $g_N$ are connected for $g_N \approx a_0$. For this we take an interpolation function $\lambda:[0,\infty)\rightarrow[0,\infty)$ such that
	\begin{equation*}
		\begin{array}{ll}
		\lambda(\sigma) \approx \sqrt{a_0}/\sqrt{\sigma} & \text{if } \sigma \ll a_0, \\
		\lambda(\sigma) \approx 0 & \text{if } \sigma \gg a_0,
		\end{array}
	\end{equation*}
	and say that $g_{real}$ and $g_N$ are related via
	\begin{equation*}
		g_{real} = g_N + \lambda\left(\left|g_N\right|\right)g_N .
	\end{equation*}
	This is, e.g., the form how the basic MOND paradigm is implemented in the frequently used N-body code Phantom of RAMSES \citep{2015CaJPh..93..232Lueghausen}. There are different suggestions about the correct choice for $\lambda$ based on observational data \citep[see, e.g., ][]{2012LRR....15...10FamaeyMcGaugh}. Here in this paper, we prove all our results under the following, very general assumptions on $\lambda$:
	
	\begin{assumptionsLambda}
		$\lambda:[0,\infty)\rightarrow[0,\infty)$ is always a measurable function and for each statement, which we prove in this paper, a selection of the follwing three assumptions shall hold:
		\begin{enumerate}[label=($\Lambda$\arabic*)]
			\item There is $\Lambda_1>0$ such that $\lambda(\sigma) \geq \Lambda_1/\sqrt{\sigma}$, for $\sigma>0$ small, \label{lambda bounded from below}
			\item There is $\Lambda_2>0$ such that $\lambda(\sigma) \leq \Lambda_2/\sqrt{\sigma}$, for every $\sigma>0$, \label{lambda bounded from above}
			\item $\lambda\in C^1((0,\infty))$, $\lambda(\sigma)\rightarrow0$ as $\sigma\rightarrow\infty$ and there is $\Lambda_2>0$ such that $-\Lambda_2/(2\sigma^{3/2}) \leq \lambda'(\sigma)\leq 0$, for $\sigma>0$. \label{lambda Prime bounded from below}
		\end{enumerate}
		We precise for each statement which assumptions are necessary for its proof.
	\end{assumptionsLambda}
	\ref{lambda bounded from below} is necessary such that accelerations due to MOND are much stronger than their Newtonian counterpart in the low acceleration regime. \ref{lambda bounded from above} takes care that the accelerations in MOND stay in their physically motivated range. \ref{lambda Prime bounded from below} is actually only a stronger version of \ref{lambda bounded from above} that grants more regularity, as we show below.
	
	Now, we would like to define
	\begin{equation} \label{definition MOND potential in spherical symmetry}
		\nabla U^M_{\rho} := \nabla U^N_\rho + \lambda\left(\left|\nabla U^N_\rho\right|\right)\nabla U^N_\rho
	\end{equation}
	and say that $U^M_\rho$ is simply the antiderivative of $\nabla U^M_\rho$. However, this is only possible if $\nabla U^M_\rho$ is an irrotational vector field. As long as $\nabla U^M_\rho$ is spherically symmetric, this is always true. Since for most parts of this paper we consider spherically symmetric situations, the above definition of $U^M_\rho$ is just fine. But if we leave spherical symmetry, the above definitions of $U^M_\rho$ and $\nabla U^M_\rho$ are wrong and one has to rely on more refined implementations of the basic MOND paradigm.
	
	Once we leave spherical symmetry we will make use of the QUMOND theory (QUasi linear formulation of MOND) \citep{2010MNRAS.403..886Milgrom}. In QUMOND the Mondian potential $U^M_\rho$ is related to $U^N_\rho$ via
	\begin{equation} \label{definition MOND potential QUMOND UM}
	U^M_\rho = U^N_\rho + U^\lambda_\rho
	\end{equation}
	where
	\begin{equation} \label{definition MOND potential QUMOND Ulambda}
		U^\lambda_\rho (x) := \frac{1}{4\pi} \int \lambda\left(|\nabla U^N_\rho(y)|\right)\nabla U^N_\rho(y) \cdot \left(\frac{x-y}{|x-y|^3} + \frac{y}{|y|^3}\right) \diff y, \quad x\in\R^3.
	\end{equation}
	In \cite{2024Frenkler} we have studied QUMOND in detail from a mathematicians point of view. In Sections \ref{section Newtonian potentials} and \ref{section Mondian potentials} we summarize several important lemmas from this work, upon which we rely subsequently.


	\subsection{Newtonian potentials} \label{section Newtonian potentials}

	The three lemmas of this section are all proven in a slightly more general version in \cite{2024Frenkler}. Therefore, we omit their proofs.

	\begin{lem}[Regularity of Newtonian potentials] \label{lemma Newtonian potential}
		If $\rho\in C_c^1(\R^3)$ then
		\begin{equation*}
			U^N_\rho(x) := - \int \frac{\rho(y)}{|x-y|} \diff y, \quad x\in\R^3,
		\end{equation*}
		is the unique solution of \eqref{Poisson equation} in $C^2(\R^3)$. Its first derivative is given by
		\begin{equation*}
			\nabla U^N_\rho(x) = \int \frac{x-y}{|x-y|^3} \rho(y) \diff y, \quad x\in\R^3.
		\end{equation*}
		We say that $\rho$ is an $L^p$ function on $\R^3$, $1<p<\infty$, if the following norm of $\rho$ is finite:
		\begin{equation*}
			\|\rho\|_p := \left(\int |\rho(x)|^p \diff x\right)^{1/p} < \infty.
		\end{equation*}
		If $\rho\in L^1\cap L^p(\R^3)$ for some $1<p<\infty$ (not necessarily differentiable) then $U^N_\rho$ defined as above is twice weakly differentiable and the above formula for $\nabla U^N_\rho$ and the following estimates hold:
		\begin{enumerate}[label=\alph*)]
			\item If $1<p<\frac{3}{2}$ and $3<r<\infty$ with $\frac{1}{3} + \frac{1}{p} = 1 + \frac{1}{r}$ then
			\begin{equation*}
			\|U^N_\rho\|_r \leq C_{p,r} \|\rho\|_p.
			\end{equation*}
			\item If $1<p<3$ and $\frac{3}{2} < s < \infty $ with $\frac{2}{3} + \frac{1}{p} = 1 + \frac{1}{s}$ then
			\begin{equation*}
			\|\nabla U^N_\rho\|_s \leq C_{p,s} \|\rho\|_p.
			\end{equation*}
			\item For every $1<p<\infty$
			\begin{equation*}
			\|D^2U^N_\rho\|_p \leq C_p \|\rho\|_p.
			\end{equation*}
		\end{enumerate}
	\end{lem}

	For studying spherically symmetric situations, it is important to recall Newtons shell theorem.
	
	\begin{lem}[Newtons shell theorem] \label{lemma Newtons shell theorem}
		Let $1<p<3$ and $\rho\in L^1\cap L^p(\R^3)$, $\geq 0$ be spherically symmetric. Then
		\begin{equation*}
		\nabla U^N_\rho(x) = \frac{M(r)}{r^2} \frac{x}{r}
		\end{equation*}
		for a.e. $x\in\R^3$ with $r= |x|$ and
		\begin{equation*}
		M(r) := \int_{B_r} \rho(y) \diff y = 4\pi \int_0^r s^2 \rho(s) \diff s
		\end{equation*}
		denoting the mass inside the ball with radius $r$.
	\end{lem}

	Later on, we will make regular use of the following statement.
	
	\begin{lem} \label{lemma EpotN for rho in L6/5}
		If $\rho,\sigma\in L^{6/5}(\R^3)$ then
		\begin{equation*}
		-\frac{1}{8\pi} \int \nabla U^N_\rho\cdot \nabla U^N_\sigma \diff x = \frac 12 \int U^N_\rho\sigma\diff x = -\frac 12 \iint \frac{\rho(y)\sigma(x)}{|x-y|} \diff x\diff y.
		\end{equation*}
	\end{lem}
	
	\subsection{Mondian potentials} \label{section Mondian potentials}
	
	The first two lemmas in this section are proven in a more general version in \cite{2024Frenkler}. Therefore, as in the previous section, we omit their proofs.
	
	Let us first take a closer look on the interpolation function $\lambda$. In the next Lemma we show that the assumption \ref{lambda Prime bounded from below} is just a stronger version of \ref{lambda bounded from above} and implies further the H\"older continuity of the function $\lambda(|u|)u$, $u\in\R^3$.
	
	\begin{lem}[Regularity of interpolation function $\lambda$] \label{lemma lambda is Hoelder continuous}
		If $\lambda:(0,\infty)\rightarrow(0,\infty)$ satisfies \ref{lambda Prime bounded from below} then it satisfies also \ref{lambda bounded from above} and there is a $C>0$ such that for all $u,v\in\R^3$
		\begin{equation*}
		|\lambda(|u|)u - \lambda(|v|)v| \leq C|u-v|^{1/2}
		\end{equation*}
		with $\lambda(|u|)u = 0$ if $u=0$.
	\end{lem}

	The Mondian potential $U^M_\rho$ that belongs to a certain Newtonian potential $U^N_\rho$ is given by
	\begin{equation*}
		U^M_\rho = U^N_\rho + U^\lambda_\rho
	\end{equation*}
	where the definition of $U^\lambda_\rho$ was already given in \eqref{definition MOND potential QUMOND Ulambda}. The operator necessary to define $U^\lambda_\rho$ is not trivial. The next Lemma asserts that $U^\lambda_\rho$ is indeed well defined.
	
	\begin{lem}[Regularity of Mondian potentials] \label{lemma Ulambda}
		Assume that \ref{lambda bounded from above} holds and let $\rho\in L^1\cap L^p(\R^3)$ for a $p>1$. Set
		\begin{equation*}
		U^\lambda_\rho (x) := \frac{1}{4\pi} \int \lambda\left(|\nabla U^N_\rho(y)|\right)\nabla U^N_\rho(y) \cdot \left(\frac{x-y}{|x-y|^3} + \frac{y}{|y|^3}\right) \diff y, \quad x\in\R^3.
		\end{equation*}
		Then
		\begin{equation*}
		U^\lambda_\rho \in L^1_{loc}(\R^3).
		\end{equation*}
	\end{lem}

	If one wants to dig deeper into the QUMOND theory, one can use the Helmholtz-Weyl decomposition \citep{2011Galdi} to prove that $U^\lambda_\rho$ is weakly differentiable with $\nabla U^\lambda_\rho$ being the irrotational part of the field $\lambda(\left|\nabla U^N_\rho\right|)\nabla U^N_\rho$ \citep{2024Frenkler}.
	
	If $U^N_\rho$ is spherically symmetric, the above definition of $U^\lambda_\rho$ becomes much simpler: If $U^N_\rho$ is spherically symmetric, then so is $\lambda\left(|\nabla U^N_\rho|\right)\nabla U^N_\rho$. Thus the rotation of this field vanishes. As just stated $\nabla U^\lambda_\rho$ is the irrotational part of $\lambda\left(|\nabla U^N_\rho|\right)\nabla U^N_\rho$ in the sense of the Helmholtz-Weyl decomposition. Since in spherical symmetry the latter field is already irrotational, the Helmholtz-Weyl decompoisition becomes trivial and
	 the two fields $\nabla U^\lambda_\rho$ and $\lambda\left(|\nabla U^N_\rho|\right)\nabla U^N_\rho$ are indeed identical \citep{2024Frenkler}:
	\begin{equation*}
		\nabla U^\lambda_\rho = \lambda\left(|\nabla U^N_\rho|\right)\nabla U^N_\rho.
	\end{equation*}
	Therefore, in spherical symmetry the following lemma holds:
	
	\begin{lem}[Regularity of Mondian potential in spherical symmetry] \label{lemma Ulambda sph sym}
		Let $\rho\in L^1\cap L^p(\R^3)$, $1<p<3$ be spherically symmetric. Then
		\begin{equation*}
			\nabla U^M_\rho = \nabla U^N_\rho + \nabla U^\lambda_\rho = \nabla U^N_\rho + \lambda\left(|\nabla U^N_\rho|\right)\nabla U^N_\rho.
		\end{equation*}
		If additionally \ref{lambda Prime bounded from below} holds, then
		\begin{equation*}
		U^M_\rho, U^N_\rho, \, U^\lambda_\rho \in C^1(\R^3\backslash\{0\})
		\end{equation*}
		and for $x\in\R^3\backslash\{0\}$, $r=|x|$ holds
		\begin{equation*}
		U^M_{\rho}(r) =  U^M_{\rho}(1) + \int_1^r \left(1 + \lambda\left(\frac{M(s)}{s^2}\right)\right) \frac{M(s)}{s^2} \diff s
		\end{equation*}
		under a slight abuse of notation.
	\end{lem}
	
	\begin{proof}
		Lemma \ref{lemma Newtons shell theorem} tells us that
		\begin{equation*}
		\nabla U^N_\rho (x) = \frac{M(r)}{r^2} \frac{x}{r}, \quad x\in\R^3\backslash\{0\},\, r=|x|.
		\end{equation*}
		Since $ \rho\in L^1(\R^3)$
		\begin{equation*}
		M(r) \in C([0,\infty)).
		\end{equation*}
		Hence
		\begin{equation*}
		\nabla U^N_\rho \in C(\R^3\backslash\{0\}).
		\end{equation*}
		Since \ref{lambda Prime bounded from below} holds, the map
		\begin{equation*}
		\R^3 \ni u \mapsto \lambda(|u|)u
		\end{equation*}
		is continuous (Lemma \ref{lemma lambda is Hoelder continuous}). As discussed above we have
		\begin{equation*}
		\nabla U^\lambda_\rho = \lambda\left(|\nabla U^N_\rho|\right)\nabla U^N_\rho.
		\end{equation*}
		in the situation of spherical symmetry.
		Thus
		\begin{equation*}
		\nabla U^\lambda_\rho  \in C(\R^3\backslash\{0\}).
		\end{equation*}
		Hence
		\begin{equation*}
		U^\lambda_\rho,\, U^N_\rho \in C^1(\R^3\backslash\{0\})
		\end{equation*}
		and by the fundamental theorem of calculus
		\begin{equation*}
		U^M_{\rho}(r) = U^N_\rho(r) + U^\lambda_\rho(r) = U^M_{\rho}(1) + \int_1^r \left(1 + \lambda\left(\frac{M(s)}{s^2}\right)\right) \frac{M(s)}{s^2} \diff s.
		\end{equation*}
		
	\end{proof}

	\subsection{Potential energy}
	
	In Newtonian physics the potential energy corresponding to a density $\rho(x)$ is given by
	
	\begin{equation*}
		\Epot^N(\rho) = - \frac{1}{8\pi} \int \left|\nabla U^N_\rho\right|^2 \diff x.
	\end{equation*}
	This is well defined for every $\rho\in L^{6/5}(\R^3)$; see Lemma \ref{lemma EpotN for rho in L6/5}. 
	
	In Mondian physics one can formally derive that the potential energy corresponding to a density $\rho(x)$ is given by
	\begin{equation*}
		\Epottilde^M (\rho) = \Epot^N(\rho) - \frac{1}{4\pi} \int\scriptQ\left( \left| \nabla U^N_\rho \right| \right) \diff x
	\end{equation*}
	where
	\begin{equation*}
		Q(v) := \int_0^v \lambda(w)w \diff w, \quad v\in[0,\infty),
	\end{equation*}
	\citep{2010MNRAS.403..886Milgrom}. But there is the problem that the integral $\int Q(\ldots) \diff x$ is in general not finite. If we choose for example a compactly supported density $\rho$ then
	\begin{equation*}
	\nabla U^N_\rho (x) = O (|x|^{-2})
	\end{equation*}
	for $|x|\rightarrow \infty$. If we choose further $\lambda(v) = 1/\sqrt v$, then
	\begin{equation*}
	\scriptQ(v) = \frac{2}{3} v^{3/2}, \quad v\geq 0,
	\end{equation*}
	and
	\begin{equation*}
	\scriptQ(\nabla U^N_\rho (x)) = O (|x|^{-3})
	\end{equation*}
	and this is not integrable. Nevertheless, we can study the difference between the potential energies of two densities $\rho$ and $\bar\rho$. If $\rho,\bar\rho\in L^1\cap L^{6/5}(\R^3)$ have the same mass and decay sufficiently fast at infinity, then	
	\begin{align*}
		\tilde E_{pot}(\rho) - \tilde E_{pot}(\bar\rho) = & - \frac{1}{8\pi} \int \left( \left| \nabla U^N_\rho \right|^{2} - \left| \nabla U^N_{\bar\rho} \right|^{2} \right) \diff x \\
		& - \frac{1}{4\pi} \int \left( \scriptQ\left( \left| \nabla U^N_\rho \right| \right) - \scriptQ\left( \left| \nabla U^N_{\bar\rho} \right| \right) \right) \diff x
	\end{align*}
	is finite -- provided $\lambda$ is sufficiently well behaved (Lemma \ref{lemma Mondian potential energy is finite} in the appendix). However, in this paper we are always interested in `the' Mondian potential energy of one specific density $\rho$. Since it only makes sense to study differences of potential energies, we fix $\bar\rho$ as a reference density. Then the (finite) term $\int|\nabla U^N_{\bar\rho}|^2\diff x$ above adds only a negligible constant. Hence we drop it and work in the sequential with the following definition:
	
	\begin{defn}[Mondian potential energy]
		Assume that \ref{lambda bounded from above} holds. Fix a reference density $\bar\rho\in C_c(\R^3)$, $\geq 0$, spherically symmetric with mass $M = \|\bar\rho\|_1 > 0$. Then for every density $\rho\in L^1\cap L^{6/5}(\R^3)$, $\geq 0$, that has the same mass $M$ as $\bar\rho$ and that is either compactly supported or decays sufficiently fast at infinity, the Mondian potential energy of $\rho$ (w.r.t the reference density $\bar\rho$) is finite; it is defined as
		\begin{equation*}
			\Epot^M(\rho) := \Epot^N(\rho) + \Epot^Q(\rho)
		\end{equation*}
		where
		\begin{align*}
			\Epot^N(\rho) & := - \frac{1}{8\pi} \int \left|\nabla U^N_\rho\right|^2 \diff x,\\
			\Epot^Q(\rho) & := - \frac{1}{4\pi} \int \left( \scriptQ\left( \left| \nabla U^N_{\rho} \right| \right) - \scriptQ\left( \left| \nabla \bar U^N \right| \right) \right)\diff x.
		\end{align*}
		We use the notation
		\begin{equation*}
			\nabla \bar U^N := \nabla U^N_{\bar\rho}.
		\end{equation*}
	\end{defn}

	When we talk about the Mondian potential energy, we refer to $\Epot^N$ as its 'Newtonian' part and to $\Epot^Q$ as its 'Mondian' part. $\Epot^Q$ is the part of the Mondian potential energy that makes Mondian physics different from Newtonian physics.

	We give a proof that $\Epot^M(\rho)$ is indeed finite, provided $\rho$ is compactly supported, in the appendix (Lemma \ref{lemma Mondian potential energy is finite}). We do not examine any further how fast $\rho$ needs to decay at infinity to have finite Mondian potential energy since we do not need this information in this paper.
	
	Last in this section let us collect everything we need to know about the regularity of the function $\scriptQ$ appearing in the definition of the Mondian potential energy.
	
	\begin{lem}[$\scriptQ$ continuous] \label{lemma Q continuous}
		Assume that \ref{lambda bounded from below} and \ref{lambda bounded from above} hold.
		Then $\scriptQ\in C([0,\infty))$ and monotonic increasing. If $u \geq v\geq 0$, it holds that
		\begin{equation*}
		\scriptQ(u)  \leq \frac{2\Lambda_2}{3} u^{3/2}
		\end{equation*}
		and
		\begin{equation*}
		\scriptQ(u) - \scriptQ(v) \leq \frac{2\Lambda_2}{3} (u^{3/2} - v^{3/2} ).
		\end{equation*}
		If $u\geq v\geq 0$ are small, it holds additionally that
		\begin{equation*}
		\scriptQ(u)  \geq \frac{2\Lambda_1}{3} u^{3/2}
		\end{equation*}
		and
		\begin{equation*}
		\scriptQ(u) - \scriptQ(v) \geq \frac{2\Lambda_1}{3} (u^{3/2} - v^{3/2} ).
		\end{equation*}
	\end{lem}
	
	\begin{proof}
		Since $\lambda$ is measurable and \ref{lambda bounded from above} guarantees the integrability of $\lambda(w)w$, $\scriptQ$ is continuous. Since $\lambda\geq 0$, $\scriptQ$ is monotonic. Further with the transformation $w = \tilde w^{2/3}$ we have
		\begin{equation*}
		\scriptQ(u) = \frac 23 \int_0^{u^{3/2}} \lambda(\tilde w^{2/3}) \tilde w^{1/3} \diff \tilde w.
		\end{equation*}
		Thanks to \ref{lambda bounded from above}, we have
		\begin{equation*}
		\lambda(\tilde w^{2/3})\tilde w^{1/3} \leq \Lambda_2.
		\end{equation*}
		Hence 
		\begin{equation*}
			\scriptQ(u) \leq \frac{2\Lambda_2}{3} u ^{3/2}
		\end{equation*}
		and for $u \geq v \geq 0$
		\begin{equation}
		\scriptQ(u) - \scriptQ(v) \leq \frac 23 \Lambda_2 ( u^{3/2} - v^{3/2} ).
		\end{equation}
		The other two estimates follow in the same way using \ref{lambda bounded from below}.
	\end{proof}

	Further we will need the following first order Taylor expansion of $\scriptQ$:
	
	\begin{lem}[$\scriptQ$ differentiable] \label{lemma Q differentiable}
		Assume that \ref{lambda Prime bounded from below} holds. Then $\scriptQ\in C^1([0,\infty))$ and there is a $C>0$ such that for all $u,v\in\R^3$
		\begin{equation*}
		\left| \scriptQ(|u|) - \scriptQ(|v|) - \lambda(|v|)v\cdot (u-v) \right| \leq C|u-v|^{3/2}.
		\end{equation*}
	\end{lem}
	
	\begin{proof}
		First observe that Lemma \ref{lemma lambda is Hoelder continuous} implies
		\begin{equation*}
			\scriptQ\in C^1([0,\infty)).
		\end{equation*}
		Now let $u,v\in\R^3$ such that
		\begin{equation*}
		w_t := v + t(u-v) \neq 0
		\end{equation*}
		for all $t\in[0,1]$. Set
		\begin{equation*}
		q(t) := \scriptQ(|w_t|).
		\end{equation*}
		Then
		\begin{equation*}
		q\in C^2([0,1])
		\end{equation*}
		with
		\begin{equation*}
		q'(t) = \scriptQ'(|w_t|)\frac{w_t}{|w_t|}\cdot(u-v) = \lambda(|w_t|)w_t\cdot(u-v)
		\end{equation*}
		and
		\begin{equation*}
		q''(t) = \lambda'(|w_t|)\frac{[w_t\cdot(u-v)]^2}{|w_t|} + \lambda(|w_t|)|u-v|^2.
		\end{equation*}
		We have
		\begin{align*}
		q(1)-q(0) & = \int_0^1 q'(s) \diff s = \int_0^1\left( q'(0) + \int_0^s q''(t) \diff t \right) \diff s \\
		& = q'(0) + \int_0^1 \int_t^1\diff s \,q''(t) \diff t \\
		& = q'(0) + \int_0^1 (1-t)q''(t) \diff t.
		\end{align*}
		Hence
		\begin{equation*}
		\scriptQ(|u|) - \scriptQ(|v|) - \lambda(|v|)v\cdot (u-v) = \int_0^1(1-t)q''(t) \diff t.
		\end{equation*}
		By \ref{lambda bounded from above} and \ref{lambda Prime bounded from below}
		\begin{equation*}
		|q''(t)| \leq C\frac{|u-v|^2}{|w_t|^{1/2}}.
		\end{equation*} 
		Hence
		\begin{equation*}
		\left| \int_0^1(1-t)q''(t) \diff t \right| \leq C|u-v|^2 \int_0^1 \frac{\diff t}{|w_t|^{1/2}}.
		\end{equation*}
		The integral on the right side becomes maximal if $w_t$ would be zero for $t=1/2$. Hence
		\begin{align*}
		\int_0^1\frac{\diff t}{|w_t|^{1/2}} &\leq \int_0^1\frac{\diff t}{\left|-\frac 12|u-v|+t|u-v|\right|^{1/2}}\\
		& = |u-v|^{-1/2} \int_0^1\frac{\diff t}{|t-1/2|^{1/2}} \\
		&= 2\sqrt 2 |u-v|^{-1/2}.
		\end{align*}
		Thus
		\begin{equation*}
		\left|\scriptQ(|u|) - \scriptQ(|v|) - \lambda(|v|)v\cdot (u-v)\right| \leq C|u-v|^{3/2}.
		\end{equation*}
	\end{proof}


	\section{Finding minimizers} \label{section finding minimizers}
	
	\subsection{Finding a minimizer in the fluid situation} \label{section existance minimizer pressure supported}
	
	\begin{generalassumptions}
		Throughout this section we assume that \ref{lambda bounded from below} and \ref{lambda bounded from above} hold.
	\end{generalassumptions}
	
	Recall the variational problem we want to solve to construct an equilibrium solution of the Euler equations \eqref{Euler equations}:
	
	We take an ansatz function $\Psi$ that satisfies the following assumptions
	
	\begin{assumptionsPsi}
		$\Psi \in C^1([0,\infty))$, $\Psi(0)=\Psi'(0)=0$ and it holds:
		\begin{enumerate}[label=($\Psi$\arabic*)]
			\item $\Psi$ is strictly convex,
			\item $\Psi(\rho)\geq C\rho^{1+1/n}$ for $\rho> 0$ large, where $0 < n < 3$.
		\end{enumerate}
	\end{assumptionsPsi}

	Let $M>0$ and fix a reference density $\bar\rho\in C_c(\R^3)$, $\geq 0$, spherically symmetric with $\|\bar\rho\|_1=M$. Let $\bar R>0$ be such that
	\begin{equation*}
	\supp \bar\rho = B_{\bar R}.
	\end{equation*}
	For a sufficiently regular density $\rho$ define the functional
	\begin{equation}
		\He(\rho) := \Epot^M(\rho) + \mathcal C(\rho) 
	\end{equation}
	where $\Epot^M(\rho)$ is the Mondian potential energy of $\rho$ w.r.t. the reference density $\bar\rho$ and
	\begin{equation*}
		\mathcal C(\rho) := \int \Psi (\rho) \diff x.
	\end{equation*}
	The subscript `E' in $\He$ indicates that we use this functional to construct a stable equilibrium solution of the Euler equations. We search a minimizer of $\He$ over the set
	\begin{equation*}
		\mathcal R_M := \left\{ \rho\in L^1(\R^3) \text{ sph. sym.}\middle| \rho\geq0, \,\|\rho\|_1 = M,\, |\Epot^Q(\rho)| + \mathcal C(\rho) < \infty \right\}.
	\end{equation*}	
	This problem we solve in the present section.

	The first step toward our goal is to find bounds for the potential energy. Here we use the spherical symmetry to get also good bounds for the Mondian part of the potential energy.
	
	\begin{lem} \label{lemma estimates for Epot}
		There are $C_0>0$ and $C_1=C_1(\bar R, \Lambda_2)>0$ such that for all $\rho\in\mathcal R_M$
		\begin{equation*}
		-\Epot^N(\rho) \leq C_0 \|\rho\|_{6/5}^2
		\end{equation*}
		and
		\begin{equation*}
		-\Epot^Q(\rho) \leq \frac{1}{4\pi} \int_{|x|\leq \bar R} \scriptQ\left(\left|\nabla U^N_\rho\right|\right) \diff x \leq - C_1 \Epot^N(\rho)^{3/4} \leq C_1 C_0 \|\rho\|_{6/5}^{3/2}.
		\end{equation*}
	\end{lem}
	
	\begin{proof}
		The inequality for $\Epot^N$ is a direct consequence of the Hardy-Littlewood-Sobolev inequality \citep[Theorem 4.3]{2010LiebLoss} since for $\rho\in L^{6/5}(\R^3)$
		\begin{equation*}
		-\Epot (\rho) = \frac{1}{2} \iint \frac{\rho(x)\rho(y)}{|x-y|} \diff x \diff y;
		\end{equation*}
		see Lemma \ref{lemma EpotN for rho in L6/5}.
		So let us study the inequality for $\Epot^Q$: $\scriptQ$ is monotonic. Thanks to spherical symmetry we have further for $|x|\geq \bar R$
		\begin{equation*}
		|\nabla \bar U^N(x)| = \frac{M}{|x|^2} \geq |\nabla U^N_\rho(x)|.
		\end{equation*}
		Thus
		\begin{align*}
		-\Epot^Q(\rho) &= \frac{1}{4\pi} \int_{|x|\leq \bar R} \left(\scriptQ\left(\left|\nabla U^N_\rho\right|\right) - \scriptQ\left(\left|\nabla \bar  U^N\right|\right)\right) \diff x + \int_{|x|>\bar R} \ldots \diff x \\
		& \leq \frac{1}{4\pi} \int_{|x|\leq \bar R} \scriptQ\left(\left|\nabla U^N_\rho\right|\right) \diff x.
		\end{align*}
		Further H\"older and Lemma \ref{lemma Q continuous} imply
		\begin{align*}
		\frac{1}{4\pi} \int_{|x|\leq \bar R} \scriptQ\left(\left|\nabla U^N_\rho\right|\right) \diff x & \leq \frac{\Lambda_2}{6\pi} \left(\int_{|x|\leq \bar R}  |\nabla U^N_\rho|^{2}  \diff x \right)^{3/4} \left( \frac{4\pi}{3} \bar R^3 \right)^{1/4} \\
		& \leq - \Lambda_2 C(\bar R) \Epot^N(\rho)^{3/4}.
		\end{align*}
		Applying the inequality for $\Epot^N(\rho)$ closes the proof.
	\end{proof}

	Next we prove that $\He$ is bounded from below on $\mathcal R_M$ and several bounds along minimizing sequences.
	
	\begin{lem} \label{lemma Hr bounded from below and bounds along minimizing sequences}
		$\inf_{\mathcal R_M} \He > -\infty$ and along every minimizing sequence $(\rho_j)\subset \mathcal R_M$ of $\He$
		\begin{equation*}
		\|\rho_j\|_{6/5}, \, \|\rho_j\|_{1+1/n} \text{ and } \int\Psi(\rho_j)\diff x
		\end{equation*}
		remain bounded.
	\end{lem}
	
	\begin{proof}
		Let $\rho\in \mathcal R_M$. From \ref{Psi of rho is large for rho large} and the definition of $\mathcal R_M$ we can deduce
		\begin{align} \label{equ bounds along minimizing sequences 1+1/n norm}
		\|\rho\|_{1+1/n}^{1+1/n} &= \int_{\{\rho\leq 1\}} \rho^{1+1/n} \diff x + \int_{\{\rho > 1 \}} \rho^{1+1/n} \diff x \nonumber\\
		& \leq M + C\int\Psi(\rho) \diff x.
		\end{align}
		Let either $\beta=2$ or $\beta=3/2$. Since for $\alpha = (n+1)/6$
		\begin{equation*}
		\frac{1-\alpha}{1} + \frac{\alpha}{1+1/n} = \frac{5-n}{6} + \frac n6 = \frac 56,
		\end{equation*}
		the interpolation formula yields
		\begin{align} \label{equ bounds along minimizing sequences 6/5 norm}
		\|\rho\|_{6/5}^\beta &\leq \|\rho\|_1^{(1-\alpha)\beta} \|\rho\|_{1+1/n}^{\alpha\beta} \nonumber\\
		& \leq C\left( 1 + \int \Psi(\rho)\diff x\right)^{n\beta/6} \nonumber\\
		& \leq C + C\left(\int\Psi(\rho)\diff x\right)^{n\beta/6};
		\end{align}
		in the last inequality we used that $n\beta/6 < 1$. Using now Lemma \ref{lemma estimates for Epot}, we can estimate $\He$ from below:
		\begin{align} \label{equ bounds along minimizing sequences int Psi(rho)}
		\He(\rho) & \geq \int\Psi(\rho)\diff x - C\|\rho\|_{6/5}^2 - C\|\rho\|_{6/5}^{3/2} \nonumber \\
		& \geq \int\Psi(\rho)\diff x - C - C\left( \int\Psi(\rho) \diff x \right)^{n/3} - C\left( \int\Psi(\rho) \diff x \right)^{n/4}.
		\end{align}
		Since $n<3$, this implies that $\He$ is bounded from below on $\mathcal R_M$:
		\begin{equation*}
		\inf_{\mathcal R_M} \He \geq \min_{a\geq 0} (a - C - Ca^{n/3} - Ca^{n/4}) > -\infty.
		\end{equation*}
		Let now $(\rho_j)\subset \mathcal R_M$ be a minimizing sequence of $\He$. Then \eqref{equ bounds along minimizing sequences int Psi(rho)} implies that $\int\Psi(\rho_j)\diff x$ is bounded. \eqref{equ bounds along minimizing sequences 1+1/n norm} and \eqref{equ bounds along minimizing sequences 6/5 norm} give the bounds for $\|\rho_j\|_{1+1/n}$ and $\|\rho_j\|_{6/5}$.
	\end{proof}

	Lemma \ref{lemma Hr bounded from below and bounds along minimizing sequences} enables us to extract from minimizing sequences $(\rho_j)$ a subsequence that converges weakly to a $\rho_0\in L^{1+1/n}$. This $\rho_0$ is our candidate for the minimizer of $\He$. To prove that it is a minimizer we have in particular to show that
	\begin{equation*}
	\Epot^M(\rho_j) \rightarrow \Epot^M(\rho_0) \quad \text{for } j \rightarrow\infty.
	\end{equation*}
	For this the following compactness result from \cite{2007Rein} is helpful.
	
	\begin{lem} \label{lemma compactness of Laplace inverse}
		Let $(\rho_j)\subset L^{1+1/n}(\R^3)$, $\geq 0$ be such that
		\begin{align*}
		& \rho_j \rightharpoonup \rho_0 \quad \text{weakly in } L^{1+1/n}(\R^3), \\
		& \forall \epsilon>0\,\exists R>0: \quad \limsup_{j\rightarrow\infty} \int_{|x|\geq R} \rho_j \diff x < \epsilon.
		\end{align*}
		Then $\nabla U^N_{\rho_j} \rightarrow \nabla U^N_{\rho_0}$ strongly in $L^2$.
	\end{lem}
	
	\begin{proof}
		See Lemma 2.5. in \cite{2007Rein}.
	\end{proof}

	Thus in order to pass with $\Epot^M$ to the limit we have to show that the mass along a minimizing sequence remains concentrated. In the Mondian situation we can prove this  more direct than in the Newtonian situation. Far away from the centre of mass Mondian forces are much higher than their Newtonian counterpart, hence they should also confine mass much more efficiently. And since $\Epot^Q$ is the term that makes the difference between Newtonian and Mondian physics, this effect should be hidden there. This turns out to be true as the following lemma shows.
	
	\begin{lem} \label{lemma masses remain concentrated along minimizing sequences}
		If $(\rho_j)\subset \mathcal R_M$ is a minimizing sequence of $\He$, then there is a $C>0$ such that for all $R>\bar R$ large enough and $j\in\N$
		\begin{equation}
		\int_{|x|>R} \rho_j \diff x \leq C\left(\log R - \log \bar R\right)^{-2/3}
		\end{equation}
	\end{lem}
	
	\begin{proof}
		For $r\geq 0$ let
		\begin{equation*}
		M(\rho_j,r) := \int_{|x|\leq r} \rho_j \diff x.
		\end{equation*}
		For $R>\bar R$ sufficiently large we can use Lemma \ref{lemma Q continuous} and get
		\begin{equation*}
		-\frac{1}{4\pi} \int_{|x|\geq \bar R} \left(\scriptQ\left(\left|\nabla U^N_{\rho_j}\right|\right) - \scriptQ\left(\left|\nabla \bar  U^N\right|\right)\right) \diff x \geq \frac{\Lambda_1}{6\pi}\int_{|x|\geq \bar R} \left( |\nabla \bar U^N|^{3/2} - |\nabla U^N_{\rho_j}|^{3/2} \right) \diff x.
		\end{equation*}
		Introducing polar coordinates and using that $M(\rho_j,r)$ takes values between $0$ and $M$ and is monotonic increasing, we can estimate further
		\begin{align*}
		-\frac{1}{4\pi} \int_{|x|\geq \bar R} \left(\scriptQ\left(\left|\nabla U^N_{\rho_j}\right|\right) - \scriptQ\left(\left|\nabla \bar  U^N\right|\right)\right) \diff x 
		& \geq \frac{2\Lambda_1}{3} \int_{\bar R}^\infty \frac{M^{3/2}-M(\rho_j,r)^{3/2}}{r} \diff r \\
		& \geq \frac{2\Lambda_1}{3} \int_{\bar R}^\infty \frac{\left(M-M(\rho_j,r)\right)^{3/2}}{r} \diff r \\
		& \geq \frac{2\Lambda_1}{3} \left(M-M(\rho_j,R)\right)^{3/2} \int_{\bar R}^R \frac{\diff r}{r}.
		\end{align*}
		Additionally, Lemma \ref{lemma estimates for Epot} and \ref{lemma Hr bounded from below and bounds along minimizing sequences} imply that $\mathcal C(\rho_j) + \Epot^N(\rho_j)$ is bounded. Hence $\Epot^Q(\rho_j)$ is bounded and
		\begin{align*}
		-\frac{1}{4\pi} \int_{|x|\geq \bar R} \left(\scriptQ\left(\left|\nabla U^N_{\rho_j}\right|\right) - \scriptQ\left(\left|\nabla \bar  U^N\right|\right)\right) \diff x  \leq \Epot^Q(\rho_j) + \frac{1}{4\pi} \int_{|x|< \bar R} \scriptQ\left(\left|\nabla U^N_{\rho_j}\right|\right) \diff x \leq C
		\end{align*}
		independent of $j\in\N$. Now we have only to use that
		\begin{equation*}
		M-M(\rho_j,R) = \int_{|x|>R} \rho_j \diff x
		\end{equation*}
		and we can close the proof.
	\end{proof}

	Now we can prove that there are minimizers of $\He$ over $\mathcal R_M$. To prove later that these minimizers are also stable against small perturbations, it is important that we prove further several convergences along minimizing sequences.
	
	\begin{thm} \label{thm existence of minimizers of Hr}
		Let $(\rho_j)\subset \mathcal R_M$ be a minimizing sequence of $\He$ Then there exists a subsequence, again denoted by $(\rho_j)$, such that
		\begin{equation*}
		\rho_j \rightharpoonup \rho_0 \quad \text{weakly in } L^{1+1/n}(\R^3).
		\end{equation*}
		$\rho_0\in\mathcal{R}_M$ is a minimizer of $\He$ and
		\begin{equation*}
		\nabla U^N_{\rho_j} \rightarrow \nabla U^N_{\rho_0} \quad \text{strongly in } L^2(\R^3).
		\end{equation*}
		If additionally $\rho_0$ is compactly supported,
		\begin{equation*}
		\nabla U^N_{\rho_j} - \nabla U^N_{\rho_0} \rightarrow 0 \quad \text{strongly in } L^{3/2}(\R^3).
		\end{equation*}
	\end{thm}
	
	\begin{rem*}
		Observe that 
		\begin{equation*}
		\nabla U^N_{\rho_0},\,\nabla U^N_{\rho_j} \notin L^{3/2}(\R^3)
		\end{equation*}
		but
		\begin{equation*}
		\nabla U^N_{\rho_j} - \nabla U^N_{\rho_0} \in L^{3/2}(\R^3).
		\end{equation*}
	\end{rem*}
	
	\begin{proof}
		Let $(\rho_j) \subset \mathcal R_M$ be a minimizing sequence of $\He$. Since by Lemma \ref{lemma Hr bounded from below and bounds along minimizing sequences} $(\rho_j)\subset L^{1+1/n}(\R^3)$ is bounded, there exists a subsequence, again denoted by $(\rho_j)$, such that
		\begin{equation*}
		\rho_j \rightharpoonup \rho_0 \quad \text{weakly in } L^{1+1/n}(\R^3).
		\end{equation*}
		First we want to prove that $\rho_0\in\mathcal R_M$. Obviously $\rho_0$ is spherically symmetric and non-negative. Due to the weak convergence and Lemma \ref{lemma masses remain concentrated along minimizing sequences}, there is for every $\epsilon>0$ an $R>0$ such that
		\begin{equation*}
		\int_{B_R} \rho_0 \diff x = \lim_{j\rightarrow\infty} \int_{B_R} \rho_j \diff x \in [M-\epsilon,M]. 
		\end{equation*}
		Thus by monotone convergence
		\begin{equation*}
		\int \rho_0 \diff x  = M.
		\end{equation*}
		Now we extract from $(\rho_j)$ a subsequence $(\hat\rho_j)$ such that
		\begin{equation*}
		\lim_{j\rightarrow\infty}\int\Psi(\hat\rho_j) \diff x = \liminf_{j\rightarrow\infty}\int\Psi(\rho_j)\diff x.
		\end{equation*}
		From Mazur's lemma we know that
		\begin{equation*}
		\forall j\in\N\, \exists N_j\geq j \text{ and } c_j^{(j)}, \ldots, c_{N_j}^{(j)} \geq 0 \text{ with } \sum_{i=j}^{N_j} c_i^{(j)} = 1
		\end{equation*}
		such that
		\begin{equation*}
		\tilde \rho_j := \sum_{i=j}^{N_j} c_i^{(j)}\hat\rho_i \rightarrow \rho_0 \quad \text{strongly in } L^{1+1/n}(\R^3).
		\end{equation*}
		We extract a subsequence $(\tilde\rho_{j_k})$ of $(\tilde\rho_j)$ such that $\tilde\rho_{j_k}$ converges to $\rho_0$ pointwise a.e.. Since $\Psi$ is continuous, $\Psi(\tilde\rho_{j_k})$ converges pointwise a.e., too. Using Fatou's lemma and the convexity of $\Psi$ we conclude that
		\begin{align} \label{equ proof existence of minimizers limsup Psi rho j}
		\int \Psi(\rho_0) \diff x& \leq \liminf_{k\rightarrow\infty} \int \Psi(\tilde\rho_{j_k}) \diff x= \liminf_{k\rightarrow\infty} \int \Psi\left( \sum_{i={j_k}}^{N_{j_k}} c_i^{({j_k})}\hat\rho_i \right) \diff x \nonumber \\
		& \leq \liminf_{k\rightarrow\infty} \sup_{l\geq j_k} \int \Psi(\hat\rho_l) \diff x = \limsup_{j\rightarrow\infty} \int\Psi(\hat\rho_j) \diff x \nonumber \\
		& = \lim_{j\rightarrow\infty} \int\Psi(\hat\rho_j) \diff x = \liminf_{j\rightarrow\infty}  \int\Psi(\rho_j) \diff x.
		\end{align}
		Lemma \ref{lemma Hr bounded from below and bounds along minimizing sequences} implies that $\liminf \int\Psi(\rho_j)\diff x<\infty$. Hence $\int\Psi(\rho_0)<\infty$. It remains to prove that $|\Epot^Q(\rho_0)|<\infty$. Then $\rho_0 \in \mathcal R_M$. We know that $\rho_0\in L^1\cap L^{1+1/n}(\R^3)$. In particular $\rho_0\in L^{6/5}(\R^3)$. Thus Lemma \ref{lemma estimates for Epot} implies
		\begin{equation*}
		\Epot^Q(\rho_0) > - \infty.
		\end{equation*}
		Before we can show that $\Epot^Q(\rho_0)<\infty$ we have to pass to the limit with the potential energy.

		
		
		For the Newtonian part this is straightforward: Lemma \ref{lemma compactness of Laplace inverse} together with Lemma \ref{lemma masses remain concentrated along minimizing sequences} imply that
		\begin{equation*}
		\nabla U^N_{\rho_j} \rightarrow \nabla U^N_{\rho_0} \quad \text{strongly in } L^2(\R^3)\text{ for }j\rightarrow\infty.
		\end{equation*}
		In particular,
		\begin{equation} \label{equ proof existence of minimizers lim Epot N rho j}
		\Epot^N(\rho_j)  \rightarrow \Epot^N(\rho_0) \quad \text{for }j\rightarrow\infty.
		\end{equation}
		Now we treat the Mondian part of the potential energy: 
		Thanks to spherical symmetry for every $\rho\in\mathcal R_M$ and $x\in\R^3$ with $|x|>\bar R$ holds
		\begin{equation} \label{equ proof existence of minimizers Q(nabla U rho) leq Q(nabla bar U) for large x}
		\scriptQ(|\nabla U^N_{\rho}(x)|) \leq \scriptQ(|\nabla \bar U^N(x)|).
		\end{equation}
		Thus for every $R>\bar R$ and $j\in\N$
		\begin{align} \label{equ proof existence of minimizers Epot(rho j) geq int over x leq R}
		\Epot^Q(\rho_j) \geq - \frac{1}{4\pi} \int_{|x|\leq R} \left( \scriptQ(|\nabla U^N_{\rho_j}|) - \scriptQ(|\nabla \bar U^N|) \right) \diff x.
		\end{align}
		Now we extract a subsequence of $(\nabla U^N_{\rho_j})$, which we denote again by $(\nabla U^N_{\rho_j})$, such that
		\begin{equation*}
		\nabla U^N_{\rho_j} \rightarrow \nabla U^N_{\rho_0} \quad \text{pointwise a.e. for } j\rightarrow\infty.
		\end{equation*}
		Since $\scriptQ$ is continuous, $\scriptQ(|\nabla U^N_{\rho_j}|)$ converges pointwise a.e., too. From Lemma \ref{lemma Q continuous} and Lemma \ref{lemma Hr bounded from below and bounds along minimizing sequences} we know further that
		\begin{equation*}
		\scriptQ(| \nabla U^N_{\rho_j} |) \leq C \frac{M(\rho_j,r)^{3/2}}{r^3} \leq C\|\rho_j\|_{6/5}^{3/2} \frac{\measure(B_r)^{1/4}}{r^3} \leq C r^{-9/4}, \quad r >0,
		\end{equation*}
		with $C>0$ independent of $j\in\N$. Hence we can apply the theorem of dominated convergence in \eqref{equ proof existence of minimizers Epot(rho j) geq int over x leq R} and it follows
		\begin{align*}
		\liminf_{j\rightarrow\infty} \Epot^Q(\rho_j) &\geq -\lim_{j\rightarrow\infty} \frac{1}{4\pi} \int_{|x|\leq R}  \scriptQ(|\nabla U^N_{\rho_j}|) \diff x +  \frac{1}{4\pi} \int_{|x|\leq R} \scriptQ(|\nabla \bar U^N|) \diff x \\
		& = -\frac{1}{4\pi}\int_{|x|\leq R} \left( \scriptQ(|\nabla U^N_{\rho_0}|) - \scriptQ(|\nabla \bar U^N|) \right) \diff x.
		\end{align*}
		\eqref{equ proof existence of minimizers Q(nabla U rho) leq Q(nabla bar U) for large x} holds for $\rho_0$, too. Thus, when we send $R\rightarrow\infty$, monotone convergence yields
		\begin{equation} \label{equ proof existence of minimizers liminf Epot Q rho j}
		\liminf_{j\rightarrow\infty} \Epot^Q(\rho_j) \geq \Epot^Q(\rho_0).
		\end{equation}
		In particular, $\Epot^Q(\rho_0) < \infty$ and thus $\rho_0\in\mathcal R_M$.
		
		Taking into account \eqref{equ proof existence of minimizers limsup Psi rho j}, \eqref{equ proof existence of minimizers lim Epot N rho j} and \eqref{equ proof existence of minimizers liminf Epot Q rho j} it follows
		\begin{align} \label{equ proof existence of minimizers passing to the limit}
		\He(\rho_0) & = \int\Psi(\rho_0) \diff x + \Epot^N(\rho_0) + \Epot^Q(\rho_0) \nonumber\\
		& \leq \liminf_{j\rightarrow\infty} \int \Psi(\rho_j) \diff x + \lim_{j\rightarrow\infty} \Epot^N(\rho_j) + \liminf_{j\rightarrow\infty} \Epot^Q(\rho_j) \nonumber\\
		& \leq \lim_{j\rightarrow\infty} \He(\rho_j) = \min_{\mathcal R_M} \He.
		\end{align}
		Since $\rho_0\in\mathcal R_M$,
		\begin{equation*}
		\He(\rho_0) = \min_{\mathcal R_M} \He
		\end{equation*}
		and $\rho_0$ is indeed a minimizer of $\He$ over $\mathcal R_M$.
		
		Now it proves important that, when we derived the estimate \eqref{equ proof existence of minimizers limsup Psi rho j}, we extracted first the subsequence $(\hat\rho_j)$. This way in the inequality \eqref{equ proof existence of minimizers passing to the limit} appears twice a $\liminf$. This actually implies that both $(\int\Psi(\rho_j)\diff x)$ and $(\Epot^Q(\rho_j))$ are convergent and
		\begin{align} \label{equ proof existence of minimizers EpotQ rhoj converges}
		\int\Psi(\rho_j) \diff x & \rightarrow \int\Psi(\rho_0) \diff x, \nonumber\\
		\Epot^Q(\rho_j) & \rightarrow \Epot^Q(\rho_0)
		\end{align}
		for $j\rightarrow\infty$. This we need when we prove next that
		\begin{equation*}
		\nabla U^N_{\rho_j} - \nabla U^N_{\rho_0} \rightarrow 0 \quad \text{strongly in } L^{3/2}(\R^3)
		\end{equation*}
		provided the support of $\rho_0$ is compact.
		
		Let us assume in the following that $\rho_0$ has compact support and let $R_0>0$ be sufficiently large such that
		\begin{equation*}
		\supp\rho_0 \subset B_{R_0}.
		\end{equation*}		
		We have already seen above that
		\begin{equation*}
		\nabla U^N_{\rho_j} \rightarrow \nabla U^N_{\rho_0} \quad \text{strongly in } L^2(\R^3)\text{ for }j\rightarrow\infty.
		\end{equation*}
		Hence
		\begin{equation*}
		\nabla U^N_{\rho_j} \rightarrow \nabla U^N_{\rho_0} \quad \text{strongly in } L^{3/2}(B_{R_0})\text{ for }j\rightarrow\infty.
		\end{equation*}
		Using spherical symmetry and the compact support of $\rho_0$ we derive the estimate
		\begin{align*}
		\left\| \nabla U^N_{\rho_0} - \nabla U^N_{\rho_j} \right\|_{L^{3/2}(\{|x|\geq R_0\})}^{3/2} & = \int_{|x|\geq R_0} \left| \nabla U^N_{\rho_0} - \nabla U^N_{\rho_j} \right|^{3/2} \diff x \\
		& = \int_{|x|\geq R_0} \frac{(M-M(\rho_j,|x|))^{3/2}}{|x|^3} \diff x \\
		& \leq \int_{|x|\geq R_0} \frac{M^{3/2}-M(\rho_j,|x|)^{3/2}}{|x|^3} \diff x \\
		& = \int_{|x|\geq R_0} \left(| \nabla U^N_{\rho_0} |^{3/2} - |\nabla U^N_{\rho_j} |^{3/2} \right) \diff x.
		\end{align*}
		For $R_0$ sufficiently large we can use Lemma \ref{lemma Q continuous} and estimate further
		\begin{align} \label{equ proof existence of minimizers estimate nabla UN0 minus nabla UNj in L3/2 on |x|>R0}
		\left\| \nabla U^N_{\rho_0} - \nabla U^N_{\rho_j} \right\|_{L^{3/2}(\{|x|\geq R_0\})}^{3/2}  \leq &\, \frac{3}{2\Lambda_1} \int_{|x|\geq R_0}\left(\scriptQ(| \nabla U^N_{\rho_0} |) - \scriptQ(|\nabla U^N_{\rho_j} |) \right) \diff x \nonumber \\
		= & \, - \frac{6\pi}{\Lambda_1} \left( \Epot^Q(\rho_0) - \Epot^Q(\rho_j) \right) \nonumber \\
		& \, - \frac{3}{2\Lambda_1} \int_{|x|< R_0}\left(\scriptQ(| \nabla U^N_{\rho_0} |) - \scriptQ(|\nabla U^N_{\rho_j} |) \right) \diff x.
		\end{align}
		Using Lemma \ref{lemma Q continuous}, the mean value theorem and H\"older
		\begin{align*}
		\frac{3}{2\Lambda_1} \int_{|x|< R_0}\left|\scriptQ(| \nabla U^N_{\rho_0} |) - \scriptQ(|\nabla U^N_{\rho_j} |) \right|  \diff x & \leq \frac{\Lambda_2}{\Lambda_1} \int_{|x|< R_0}\left| | \nabla U^N_{\rho_0} |^{3/2} - |\nabla U^N_{\rho_j} |^{3/2} \right| \diff x \\
		& \leq \frac{\Lambda_2}{\Lambda_1} \int_{|x|< R_0} \frac 32 \left( |\nabla U^N_{\rho_0}| + |\nabla U^N_{\rho_j}| \right) ^{1/2} \left| \nabla U^N_{\rho_0} - \nabla U^N_{\rho_j} \right| \diff x \\
		& \leq 
		C\left\|\nabla U^N_{\rho_0} - \nabla U^N_{\rho_j} \right\|_{L^{3/2}(B_{R_0})}
		\end{align*}
		and this converges to zero for $j\rightarrow\infty$. Together with \eqref{equ proof existence of minimizers EpotQ rhoj converges} and \eqref{equ proof existence of minimizers estimate nabla UN0 minus nabla UNj in L3/2 on |x|>R0} this implies
		\begin{equation*}
		\left\| \nabla U^N_{\rho_0} - \nabla U^N_{\rho_j} \right\|_{L^{3/2}(\{|x|\geq R_0\})} \rightarrow 0 \quad \text{for } j\rightarrow\infty.
		\end{equation*}
		So
		\begin{equation*}
		\nabla U^N_{\rho_0} - \nabla U^N_{\rho_j} \rightarrow 0 \quad \text{in } L^{3/2}(\R^3) \text{ for } j\rightarrow\infty
		\end{equation*}
		provided that $\rho_0$ has compact support.
	\end{proof}

	\subsection{Euler-Lagrange equation for a minimizer in the fluid situation} \label{section Euler Lagrange equation pressure supported}
	
	In the previous section, we have proven that there is a minimizer $\rho_0$ of the variational problem
	\begin{equation*}
		\text{minimize } \He(\rho) \text{ s.t. } \rho \in \mathcal{R}_M.
	\end{equation*}
	In this section we derive the Euler-Lagrange equation that belongs to the minimizer $\rho_0$. This shows that $\rho_0$ is indeed of the form \eqref{ansatz equilibrium solution fluid models} as stated in the introduction, more precisely
	\begin{equation*}
		\rho_0(x) =
		\begin{cases}
			(\Psi')^{-1}(E_0 - U^M_{\rho_0}(x)), & \text{if } U^M_{\rho_0} < E_0, \\
			0, & \text{if } U^M_{\rho_0} \geq E_0,
		\end{cases}
	\end{equation*}
	for some cut-off energy $E_0\in\R$. It will be important that $\lambda(|u|)u$, $u\in\R^3$, is H\"older continuous. Therefore we strengthen the general assumptions of the previous section.
	
	\begin{generalassumptions}
		Throughout this section we assume that \ref{lambda bounded from below} and \ref{lambda Prime bounded from below} hold.
	\end{generalassumptions}
	
	First a technical proposition:
	
	\begin{prop} \label{prop nabla UN phi is in every Lp if int phi is zero}
		Let $\phi\in L^\infty(\R^3)$ with compact support and $\int\phi\diff x = 0$. Then
		\begin{equation*}
		\nabla U^N_\phi \in L^p(\R^3)
		\end{equation*}
		for every $1<p\leq \infty$.
	\end{prop}
	
	\begin{proof}
		Let $\supp \phi \subset B_R$ and $x\in \R^3$, then
		\begin{equation*}
		|\nabla U^N_\phi(x)| \leq \|\phi\|_\infty \int_{|y|\leq R} \frac{\diff y}{|x-y|^2} \leq 4\pi R\|\phi\|_\infty < \infty.
		\end{equation*}
		Thus $\nabla U^N_\phi\in L^\infty(\R^3)$. Let now $|x|>2R$. Without loss of generality we may assume that
		\begin{equation*}
		x=(x_1,0,0) \in \R^3 \text{ with } x_1>2R.
		\end{equation*}
		Then
		\begin{align*}
		|\partial_{x_2}U^N_\phi(x)| & = 
		\left| \int\frac{-y_2}{|x-y|^3}\phi(y)\diff y \right| \leq
		\|\phi\|_\infty R\int_{|y|<R}\frac{\diff y}{|x-y|^3} \\
		& \leq \|\phi\|_\infty R\measure(B_R) \left(\frac{2}{|x|}\right)^3.
		\end{align*}
		And with the same calculation for $|\partial_{x_3} U_\phi(x)|$ we get that both
		\begin{equation*}
		\partial_{x_2}U^N_\phi(x), \, \partial_{x_3}U^N_\phi(x) = O(|x|^{-3})
		\end{equation*}
		for $x_1\rightarrow\infty$. Now let
		\begin{equation*}
		\alpha := \phi 1_{\{\phi>0\}}
		\end{equation*}
		and
		\begin{equation*}
		\beta := -\phi 1_{\{\phi<0\}}.
		\end{equation*}
		Then $\alpha,\beta\geq 0$, $\|\alpha\|_1=\|\beta\|_1 = \|\phi\|_1/2$ and
		\begin{equation*}
		\partial_{x_1} U^N_\phi(x) =  \int_{|y|<R}\frac{x_1-y_1}{|x-y|^3}\alpha(y)\diff y - \int_{|y|<R}\frac{x_1-y_1}{|x-y|^3}\beta(y)\diff y 
		\end{equation*}
		where both integrands on the right side are non-negative.
		
		With the same geometric arguments as in Proposition \ref{prop asymptotic behaviour of gradient UN} in the appendix, we get
		\begin{align} \label{equ proof nabla UN phi is in every Lp if int phi is zero}
		\frac{|\partial_{x_1} U^N_\phi(x)|}{\|\alpha\|_1}  \leq& \,\frac{1}{(|x|-R)^2} - \frac{\sqrt{1-R^2/|x|^2}}{(|x|+R)^2} = \nonumber\\
		= & \, \frac{1-\sqrt{1-R^2/|x|^2}}{(|x|-R)^2} \nonumber\\
		& \, + \sqrt{1-R^2/|x|^2}\left( \frac{1}{(|x|-R)^2} - \frac{1}{(|x|+R)^2} \right).
		\end{align}
		Since
		\begin{equation*}
		\frac{\diff}{\diff\sigma}\sqrt{1-\sigma} = -\frac{1}{2} (1-\sigma)^{-1/2}, \quad \sigma < 1, 
		\end{equation*}
		we have
		\begin{equation*}
		1-\sqrt{1-\frac{R^2}{|x|^2}} \leq \frac{1}{2} \left(1-\frac{R^2}{|x|^2}\right)^{-1/2}\frac{R^2}{|x|^2} = \frac{1}{2}(|x|^2-R^2)^{-1/2} \frac{R^2}{|x|} \leq \frac{C}{|x|}.
		\end{equation*}
		Further
		\begin{equation*}
		\frac{1}{(|x|-R)^2} - \frac{1}{(|x|+R)^2} \leq \frac{2}{(|x|-R)^3}(|x|+R-|x|+R) = \frac{4R}{(|x|-R)^3}.
		\end{equation*}
		Thus \eqref{equ proof nabla UN phi is in every Lp if int phi is zero} implies that
		\begin{equation*}
		\partial_{x_1}U^N_\phi(x) = O(|x|^{-3})
		\end{equation*}
		for $x_1\rightarrow\infty$. Thus in general
		\begin{equation*}
		\nabla U^N_\phi(x) = O(|x|^{-3})
		\end{equation*}
		for $|x|\rightarrow\infty$. And since $\nabla U^N_\phi$ is bounded, this implies that
		\begin{equation*}
		\nabla U^N_\phi \in L^p(\R^3)
		\end{equation*}
		for every $1<p\leq \infty$.
		
	\end{proof}

	The most delicate part in deriving the Euler-Lagrange equation belonging to $\rho_0$ is the derivative of $\Epot^Q(\rho)$. This derivative we study in the next lemma.
	
	\begin{lem} \label{lemma derivative of EpotQ}
		Let $\rho\in L^1\cap L^p(\R^3)$ for a $1<p<3$ and let $\phi\in L^\infty(\R^3)$ with compact support and $\int\phi\diff x = 0$. Then
		\begin{equation*}
		\int\left| \scriptQ\left(|\nabla U^N_{\rho+\tau\phi}|\right) - \scriptQ\left(|\nabla U^N_{\rho}|\right) \right| \diff x < \infty
		\end{equation*}
		and
		\begin{equation*}
		\lim_{\tau\rightarrow 0} - \frac{1}{\tau}\frac{1}{4\pi}\int\left( \scriptQ\left(|\nabla U^N_{\rho+\tau\phi}|\right) - \scriptQ\left(|\nabla U^N_{\rho}|\right) \right) \diff x = \int U^\lambda_\rho \phi \diff x.
		\end{equation*}
	\end{lem}
	
	\begin{rem}
		Observe that in Lemma \ref{lemma derivative of EpotQ} we did neither assume that $\rho$ is spherically symmetric (like in Lemma \ref{lemma estimates for Epot}) nor that the support of $\rho$ is compact (like in Lemma \ref{lemma Mondian potential energy is finite}). Nevertheless the difference between the Mondian part of the potential energy of $\rho$ and $\rho+\tau\phi$ is finite. The reason is that the difference between the two densities under consideration is given by $\tau\phi$ and we have made several regularity assumptions on $\phi$.
	\end{rem}
	
	\begin{proof}[Proof of Lemma \ref{lemma derivative of EpotQ}]
		We prove first the intermediate assertion
		\begin{align} \label{equ proof derivative of EpotQ intermediate assertion}
		\lim_{\tau\rightarrow 0}- \frac{1}{\tau}\frac{1}{4\pi}\int\left( \scriptQ\left(|\nabla U^N_{\rho+\tau\phi}|\right) -  \scriptQ\left(|\nabla U^N_{\rho}|\right) \right) \diff x = -\frac{1}{4\pi} \int\lambda(|\nabla U^N_\rho|) \nabla U^N_\rho \cdot \nabla U^N_\phi \diff x.
		\end{align}
		To avoid lengthy equations we use the abbreviation
		\begin{equation*}
		F(v) := \lambda(|v|)v, \quad v\in\R^3.
		\end{equation*}
		Since we assume that \ref{lambda Prime bounded from below} holds, we have more regularity for $\scriptQ$ and $\lambda$ than in the previous section. From Lemma \ref{lemma Q differentiable} follows
		\begin{equation*}
		\scriptQ \in C^1([0,\infty))
		\end{equation*}
		with
		\begin{equation*}
		\scriptQ'(\sigma) = \lambda(\sigma)\sigma, \quad \sigma\geq 0.
		\end{equation*}
		Since $\scriptQ'(0) = 0$,
		\begin{equation*}
		\scriptQ(|\cdot|) \in C^1(\R^3)
		\end{equation*}
		with
		\begin{equation*}
		\nabla \scriptQ(|v|) = F(v), \quad v\in\R^3.
		\end{equation*}
		For $u,v\in\R^3$ set
		\begin{equation*}
		f_{u,v}(t) := \scriptQ(|tu + (1-t)v|), \quad 0\leq t \leq 1.
		\end{equation*}
		From the mean value theorem follows that for every $u,v\in\R^3$ exists $s\in[0,1]$ such that
		\begin{align*}
		\scriptQ(|u|)-\scriptQ(|v|) &= f_{u,v}(1) - f_{u,v}(0) = f_{u,v}'(s) \\
		& = F(su+(1-s)v)\cdot (u-v).
		\end{align*}
		Treat $\nabla U^N_{\rho+\tau\phi}$ and $\nabla U^N_\rho$ as pointwise defined functions and interpret $\nabla U^N_{\rho+\tau s(y)\phi}$ and $\nabla U^N_\phi$ as linear combinations of the two. Then for every $y\in\R^3$ there is $0\leq s(y) \leq 1$ such that
		\begin{align*}
		\scriptQ\left(|\nabla U^N_{\rho+\tau\phi}(y)|\right) - \scriptQ\left(|\nabla U^N_{\rho}(y)|\right) = \tau F\left( \nabla U^N_{\rho+\tau s(y)\phi}(y) \right) \cdot \nabla U^N_\phi(y).
		\end{align*}
		In view of the intermediate assertion \ref{equ proof derivative of EpotQ intermediate assertion}, which we want to prove, we estimate (suppressing the $y$-argument)
		\begin{align*}
		\left| \scriptQ\left(|\nabla U^N_{\rho+\tau\phi}|\right) - \scriptQ\left(|\nabla U^N_{\rho}|\right) - \tau F(\nabla U^N_\rho)\cdot \nabla U^N_\phi \right| &\leq \tau \left| F\left( \nabla U^N_{\rho+\tau s\phi} \right) - F(\nabla U^N_\rho)\right| |\nabla U^N_\phi| \\
		& \leq C\tau \left| \nabla U^N_{\rho+\tau s\phi} - \nabla U^N_\rho \right|^{1/2} |\nabla U^N_\phi| \\
		& \leq C\tau^2 |\nabla U^N_\phi|^{3/2};
		\end{align*}
		we have used the H\"older continuity of $F(u)=\lambda(|u|)u$ (Lemma \ref{lemma lambda is Hoelder continuous}). Thus
		\begin{align*}
		\frac{1}{4\pi} \frac{1}{\tau}\int\left| \scriptQ\left(|\nabla U^N_{\rho+\tau\phi}|\right) - \scriptQ\left(|\nabla U^N_{\rho}|\right) - \tau F(\nabla U^N_\rho)\cdot \nabla U^N_\phi \right| \diff y \leq C\tau \int |\nabla U^N_\phi|^{3/2} \diff y.
		\end{align*}
		With Proposition \ref{prop nabla UN phi is in every Lp if int phi is zero} this yields the intermediate assertion \eqref{equ proof derivative of EpotQ intermediate assertion}. Further this guarantees that the integral
		\begin{equation*}
		\int\left| \scriptQ\left(|\nabla U^N_{\rho+\tau\phi}|\right) - \scriptQ\left(|\nabla U^N_{\rho}|\right) \right| \diff x < \infty
		\end{equation*}
		for all $\tau \geq 0$.
		Using the estimates from the proof of Lemma \ref{lemma Ulambda} allows us to use Fubini in the following calculation:
		\begin{align*}
		-\frac{1}{4\pi} \int \lambda(|\nabla U^N_\rho|)\nabla U^N_\rho \cdot \nabla U^N_\phi \diff x 
		& = - \frac{1}{4\pi} \int \lambda(|\nabla U^N_\rho(y)|) \nabla U^N_\rho(y) \cdot \int \frac{y-x}{|y-x|^3} \phi(x)\diff x \diff y \\
		& = \frac{1}{4\pi} \iint  \lambda(|\nabla U^N_\rho(y)|) \nabla U^N_\rho(y) \cdot \left(\frac{x-y}{|x-y|^3} + \frac{y}{|y|^3} \right) \phi(x) \diff y \diff x \\
		& = \int U^\lambda_\rho(x) \phi(x) \diff x.
		\end{align*}		
	\end{proof}

	We prove that minimizers $\rho_0$ of $\He$ satisfy the following Euler-Lagrange equation.
	
	\begin{thm} \label{thm Euler Lagrange equation}
		Let $\rho_0\in\mathcal R_M$ be a minimizer of $\He$ over $\mathcal R_M$. Then there is an $E_0\in\R$ such that for a.e. $x\in\R^3$
		\begin{equation*}
		\rho_0(x) = \begin{cases}
		(\Psi')^{-1}(E_0 -U^M_{\rho_0}(x)), & \text{if } U^M_{\rho_0}<E_0, \\
		0, & \text{if } U^M_{\rho_0}\geq E_0.
		\end{cases}
		\end{equation*}
	\end{thm}

	\begin{rem}
		As the following formal calculation shows, the above Euler-Lagrange equation implies that a minimizer $\rho_0$ of $\He$ over $\mathcal R_M$ equipped with the velocity field $u_0=0$ is an equilibrium solution of the Euler equations \eqref{Euler equations}.
		
		Since the velocity field $u_0 = 0$ and both $u_0$ and $\rho_0$ are time independent, the only thing we have to verify is
		\begin{equation*}
			\nabla p + \rho_0 \nabla U^M_{\rho_0} = 0
		\end{equation*}
		where
		\begin{equation*}
			p(x) = \rho_0\Psi'(\rho_0) - \Psi(\rho_0)
		\end{equation*}
		is the corresponding equation of state. The above equation holds on the set $\{\rho_0=0\}$. On the complementary set $\{\rho_0> 0\}$ we have
		\begin{equation*}
			\nabla p = \rho_0 \Psi''(\rho_0)\nabla \rho_0 = \rho_0 \nabla \left[ \Psi'(\rho_0) \right].
		\end{equation*}
		Using the above Euler-Lagrange equation gives
		\begin{equation*}
			\nabla p = \rho_0 \nabla \left[ E_0 - U^M_{\rho_0} \right] = -\rho_0 \nabla U^M_{\rho_0}
		\end{equation*}
		as desired.
	\end{rem}
	
	\begin{proof}[Proof of Theorem \ref{thm Euler Lagrange equation}]
		Consider $\rho_0$ as a pointwise defined function. Let $\epsilon>0$ and set
		\begin{equation*}
		S_\epsilon := \{ x\in B_{1/\epsilon} | \epsilon\leq \rho_0(x)\leq 1/\epsilon\}.
		\end{equation*}
		For $\epsilon>0$ small enough
		\begin{equation*}
		\measure(S_{\epsilon}) > 0.
		\end{equation*}
		Let $w\in L^\infty(\R^3)$ be spherically symmetric and compactly supported such that
		\begin{equation} \label{equ proof Euler Lagrange equation w geq 0}
		w\geq 0 \text{ on } \{\rho = 0 \}
		\end{equation}
		and 
		\begin{equation} \label{equ proof Euler Lagrange equation w = 0}
		w \text{ vanishes on } \supp \rho_0\backslash S_\epsilon.
		\end{equation}
		Define
		\begin{equation*}
		\phi := w - \frac{\int w\diff x}{\measure(S_\epsilon)} 1_{S_\epsilon}
		\end{equation*}
		and
		\begin{equation*}
		\rho_\tau := \rho_0 + \tau \phi, \quad \tau>0.
		\end{equation*}
		Observe that $\phi\in L^\infty(\R^3)$, $\supp\phi\subset S_\epsilon \cup \supp w$ is compact and $\int\phi\diff x = 0$.
		
		We show that $\rho_\tau \in \mathcal R_M$ for $\tau>0$ small: Obviously $\rho_\tau$ is spherically symmetric. If $\tau\geq 0$ is small, then $\rho_\tau\geq 0$ due to \eqref{equ proof Euler Lagrange equation w geq 0}, \eqref{equ proof Euler Lagrange equation w = 0} and the boundedness of $\phi$ on $S_\epsilon$. Since $\int\phi\diff x = 0$, $\int\rho_\tau\diff x = M$. Further
		\begin{align*}
		\int\Psi(\rho_\tau)\diff x &\leq \int\Psi(\rho_0) \diff x + \int_{S_\epsilon} \Psi(\rho_0+\tau\phi) \diff x  + \int_{\supp w \backslash \supp \rho_0} \Psi(w) \diff x \\
		& \leq \int\Psi(\rho_0)\diff x + \Psi(1/\epsilon+\tau\|\phi\|_\infty)\measure(B_{1/\epsilon})  + \Psi(\|w\|_\infty) \measure(\supp w) < \infty.
		\end{align*}
		And last, $\rho_0\in\mathcal{R}_M$ and Lemma \ref{lemma derivative of EpotQ} imply
		\begin{align*}
		\int \left| \scriptQ(|\nabla U^N_{\rho_\tau}|) - \scriptQ(|\nabla \bar U^N|) \right| \diff x \leq  & \, \int \left| \scriptQ(|\nabla U^N_{\rho_\tau}|) - \scriptQ(|\nabla  U^N_{\rho_0}|) \right| \diff x \\
		& \, + \int \left| \scriptQ(|\nabla U^N_{\rho_0}|) - \scriptQ(|\nabla \bar U^N|) \right| \diff x < \infty.
		\end{align*}
		Thus $\rho_\tau\in \mathcal{R}_M$.
		
		Lemma \ref{lemma derivative of EpotQ} implies further
		\begin{equation*}
		\frac{1}{\tau}(\Epot^Q(\rho_\tau) - \Epot^Q(\rho_0)) \rightarrow \int U^\lambda_{\rho_0}\phi\diff x \quad \text{for }\tau\rightarrow 0.
		\end{equation*}
		Since
		\begin{equation*}
		\frac{1}{\tau} (\Psi(\rho_\tau) - \Psi(\rho_0)) \rightarrow \Psi'(\rho_0)\phi \quad \text{pointwise for }\tau\rightarrow 0
		\end{equation*}
		and
		\begin{equation*}
		\frac{1}{\tau} |\Psi(\rho_\tau)-\Psi(\rho_0)| \leq \Psi'(1/\epsilon+\|\phi\|_\infty)|\phi|, \quad 0\leq \tau \leq 1,
		\end{equation*}
		dominated convergence implies
		\begin{equation*}
		\frac{1}{\tau}\int\left(\Psi(\rho_\tau)-\Psi(\rho_0)\right)\diff x \rightarrow \int\Psi'(\rho_0)\phi \diff x \quad \text{for }\tau\rightarrow 0.
		\end{equation*}
		Since both $\rho_\tau,\rho_0 \in \mathcal{R}_M \subset L^{6/5}(\R^3)$, we have thanks to Lemma \ref{lemma EpotN for rho in L6/5}
		\begin{align*}
		\frac{1}{\tau}(\Epot^N(\rho_\tau) - & \Epot^N(\rho_0))\\
		&= - \frac{1}{2\tau}\iint \frac{\rho_\tau(x)\rho_\tau(y)}{|x-y|} \diff x \diff y + \frac{1}{2\tau}\iint \frac{\rho_0(x)\rho_0(y)}{|x-y|} \diff x \diff y \\
		& = - \iint \frac{\phi(x)\rho(y)}{|x-y|}\diff x \diff y - \frac \tau 2 \iint \frac{\phi(x)\phi(y)}{|x-y|} \diff x \diff y.
		\end{align*}
		Thus
		\begin{equation*}
		\frac{1}{\tau}(\Epot^N(\rho_\tau) - \Epot^N(\rho_0)) \rightarrow \int U^N_{\rho_0}\phi\diff x \quad\text{for }\tau\rightarrow 0.
		\end{equation*}
		Since $\rho_0,\rho_\tau\in\mathcal R_M$ and $\rho_0$ is a minimizer of $\He$ over $\mathcal R_M$, we have
		\begin{align*}
		0 & \leq \lim_{\tau\searrow 0}\frac{1}{\tau}(\He(\rho_\tau) - \He(\rho_0)) \\
		& = \int (U^M_{\rho_0} + \Psi'(\rho_0))\phi\diff x \\
		& = \int (U^M_{\rho_0} + \Psi'(\rho_0))w\diff x - \int_{S_\epsilon} (U^M_{\rho_0} + \Psi'(\rho_0)) \frac{\int w \diff x}{\measure(S_\epsilon)} \diff y \\
		& = \int \left[ U^M_{\rho_0} + \Psi'(\rho_0) - \frac{\int_{S_\epsilon}(U^M_{\rho_0} + \Psi'(\rho_0)) \diff y}{\measure(S_\epsilon)} \right] w \diff x \\
		& = \int \left[ U^M_{\rho_0} + \Psi'(\rho_0) - E_\epsilon \right] w \diff x
		\end{align*}
		with
		\begin{equation*}
		E_\epsilon := \frac{\int_{S_\epsilon}(U^M_{\rho_0} + \Psi'(\rho_0)) \diff y}{\measure(S_\epsilon)}.
		\end{equation*}
		$w$ was arbitrary. In view of \eqref{equ proof Euler Lagrange equation w geq 0} and \eqref{equ proof Euler Lagrange equation w = 0} the above inequality implies
		\begin{align*}
		& U^M_{\rho_0} + \Psi'(\rho_0) \geq E_\epsilon && \text{a.e. on }\{\rho_0 = 0\},\\
		& U^M_{\rho_0} + \Psi'(\rho_0) = E_\epsilon && \text{a.e. on }S_\epsilon.
		\end{align*}
		Since $\epsilon>0$ was arbitrary, too, $E_\epsilon=E_0$ is independent of $\epsilon$ and
		\begin{align*}
		& U^M_{\rho_0} + \Psi'(\rho_0) \geq E_0 && \text{a.e. on }\{\rho_0 = 0\},\\
		& U^M_{\rho_0} + \Psi'(\rho_0) = E_0 && \text{a.e. on }\{\rho_0>0\}.
		\end{align*}
		Since $\Psi'(0)=0$ and $\Psi'(\sigma)>0$ if $\sigma>0$, it holds for a.e. $x\in\R^3$
		\begin{equation*}
		\rho_0(x) = 0 \iff U^M_{\rho_0}(x) \geq E_0
		\end{equation*}
		and
		\begin{equation*}
		\rho_0(x) > 0 \iff U^M_{\rho_0} < E_0.
		\end{equation*}		
		If $U^M_{\rho_0} < E_0$ the convexity of $\Psi$ gives
		\begin{equation*}
		\rho_0(x) = (\Psi')^{-1}(E_0-U^M_{\rho_0}(x)).
		\end{equation*}
	\end{proof}

	Now we have proven the desired Euler-Lagrange equation for our minimizer $\rho_0$. But we do not stop here. We continue on and use this equation to find out more about the regularity of $\rho_0$.
	
	Since for a suitable density $\rho$ $\nabla U^M_\rho(x) = O(|x|^{-1})$ for $|x|\rightarrow\infty$, the potential $U^M_\rho(x)$ diverges logarithmically when $|x|\rightarrow\infty$. Combining this with the just proven Euler-Lagrange equation yields that minimizers of $\He$ are compactly supported.
	
	\begin{lem} \label{lemma minimizers have compact support}
		Let $\rho_0\in\mathcal R_M$ be a minimizer of $\He$ over $\mathcal R_M$. Then $\rho_0$ is compactly supported.
	\end{lem}
	
	\begin{proof}
		Since $\rho_0$ is spherically symmetric and $\rho_0\in \mathcal R_M\subset L^{6/5}(\R^3)$, Lemma \ref{lemma Ulambda sph sym} states that $U^M_{\rho_0}\in C^1(\R^3\backslash \{0\})$ and
		\begin{equation*}
		U^M_{\rho_0}(r) = U^M_{\rho_0}(1) + \int_1^r \left(1+\lambda\left(\frac{M(s)}{s^2}\right)\right)\frac{M(s)}{s^2} \diff s, \quad r=|x|,\, x\in\R^3\backslash\{0\}.
		\end{equation*}
		Let $R>1$ be sufficiently large such that we can use \ref{lambda bounded from below} in the following estimate and such that $M(r) \geq M/2$ for every $r>R$. Then we get
		\begin{equation*}
		U^M_{\rho_0}(r) \geq U^M_{\rho_0}(1) + \Lambda_1 \sqrt{\frac M2} \int_R^r \frac{\diff s}{s} \geq C' + C \log r, \quad r> R,
		\end{equation*}
		with $C'\in\R$ and $C>0$.
		In particular $U^M_{\rho_0}(r)>E_0$ for $r>R$ sufficiently large. Thus $\rho_0$ is compactly supported.
	\end{proof}

	Last in this section we use the Euler-Lagrange equation to prove that minimizers are continuous.
	
	\begin{lem} \label{lemma minimizers are continuous}
		Let $\rho_0\in\mathcal R_M$ be a minimizer of $\He$ over $\mathcal R_M$. Then
		\begin{equation*}
		\rho_0 \in C_c(\R^3)
		\end{equation*}
		and
		\begin{equation*}
		U^M_{\rho_0} \in C^1(\R^3).
		\end{equation*}
	\end{lem}
	
	\begin{proof}
		As in the proof of Lemma \ref{lemma Hr bounded from below and bounds along minimizing sequences}
		\begin{equation*}
		\rho_0\in L^{1+1/n}(\R^3).
		\end{equation*}
		Hence
		\begin{equation} \label{equ proof minimizers are continuous Mr leq r to the 3/(n+1)}
		M(r) \leq \|\rho_0\|_{1+1/n}\|1_{B_r}\|_{n+1} \leq C r^{3/(n+1)}, \quad r\geq 0.
		\end{equation}
		Since $0<n<3$, we have in particular
		\begin{equation} \label{equ proof minimizers are continuous Mr leq r to the 3/4}
		M(r) \leq Cr^{3/4}, \quad 0\leq r \leq R_0,
		\end{equation}
		where $R_0>0$ is such that
		\begin{equation*}
		\supp\rho_0 = B_{R_0}.
		\end{equation*}
		\eqref{equ proof minimizers are continuous Mr leq r to the 3/4} and \ref{lambda bounded from above} imply
		\begin{equation*}
		|\nabla U^\lambda_{\rho_0}(x)| = \lambda\left(\frac{M(r)}{r^2}\right)\frac{M(r)}{r^2} \leq \frac{\sqrt{M(r)}}{r} \leq C r^{-5/8}, \quad 0<r=|x|<R_0.
		\end{equation*}
		Thus $(U^\lambda_{\rho_0}(r))'$ is in $L^1([0,R_0])$ and hence
		\begin{equation*}
		U^\lambda_{\rho_0} \in C(B_{R_0}).
		\end{equation*}
		and with Lemma \ref{lemma Ulambda sph sym}
		\begin{equation*}
		U^\lambda_{\rho_0} \in C(\R^3).
		\end{equation*}
		\eqref{equ proof minimizers are continuous Mr leq r to the 3/(n+1)} implies
		\begin{equation*}
		|\nabla U^N_{\rho_0}(x)| = \frac{M(r)}{r^2} \leq Cr^{3/(n+1)-2}, \quad r=|x|>0.
		\end{equation*}
		
		If $0<n<2$, $(U^N_{\rho_0}(r))'\in L^1([0,R_0])$ and hence
		\begin{equation*}
		U^N_\rho\in C(\R^3).
		\end{equation*}
		Then $\rho\in C_c(\R^3)$. Since \ref{lambda Prime bounded from below} holds, $\nabla U^M_{\rho_0}\in C(\R^3)$ and $U^M_{\rho_0}\in C^1(\R^3)$.
		
		If however $2\leq n<3$, some more arguments are necessary to get the same regularity for $\rho_0$ and $U^M_{\rho_0}$. Let us employ again \eqref{equ proof minimizers are continuous Mr leq r to the 3/4}. Then we have for $0<r\leq R_0$
		\begin{equation} \label{equ proof minimizers are continuous UN rho0 leq r to the -1/4}
		|U^N_{\rho_0}(r)|\leq |U^N_{\rho_0}(R_0)| + C\int_r^{R_0}s^{-5/4}\diff s \leq C r^{-1/4}.
		\end{equation}
		Observe that due to the mean value theorem, the convexity of $\Psi$ and \ref{Psi of rho is large for rho large}, holds:
		\begin{equation*}
		\Psi'(\rho) \geq \Psi'(\tau) = \frac{\Psi(\rho)-\Psi(0)}{\rho-0} = \frac{\Psi(\rho)}{\rho} \geq C\rho^{1/n}
		\end{equation*}
		for every $\rho>0$ large with an intermediate value $0<\tau<\rho$. Thus
		\begin{equation*}
		(\Psi')^{-1}(\eta) \leq C\eta^n \quad \text{for }\eta>0 \text{ large}.
		\end{equation*}
		Thus
		\begin{align*}
		\int \rho_0^4\diff x & \leq \int_{\{\rho_0 \text{ small}\}} \rho_0^4\diff x + C\int_{\{\rho_0 \text{ large}\}}(E_0-U^\lambda_{\rho_0}-U^N_{\rho_0})^{4n} \diff x \\
		& \leq C\int\rho_0 \diff x + C\int_{B_{R_0}}(1+|U^N_{\rho_0}(x)|)^{4n} \diff x.
		\end{align*}
		Using \eqref{equ proof minimizers are continuous UN rho0 leq r to the -1/4} gives
		\begin{equation*}
		\int\rho_0^4 \diff x \leq C\left(1+\int_{B_{R_0}} |x|^{-n} \diff x \right) < \infty
		\end{equation*}
		since $0<n<3$. Hence
		\begin{equation*}
		M(r) \leq \|\rho_0\|_4 \|1_{B_r}\|_{4/3} \leq Cr^{9/4}
		\end{equation*}
		and
		\begin{equation*}
		|\nabla U^N_{\rho_0}| = \frac{M(r)}{r^2} \leq Cr^{1/4}, \quad x\in\R^3, \, r=|x|.
		\end{equation*}
		Thus
		\begin{equation*}
		\nabla U^N_{\rho_0},\,\nabla U^\lambda_{\rho_0}  \in C(\R^3).
		\end{equation*}
		Hence
		\begin{equation*}
		U^M_{\rho_0}\in C^1(\R^3)
		\end{equation*}
		and
		\begin{equation*}
		\rho_0 \in C_c(\R^3).
		\end{equation*}
		
	\end{proof}

	\subsection{Constructing a minimizer in the collisionless situation} \label{section existance of minimizer in collisionless situation}
	
	\begin{generalassumptions}
		Throughout this section we assume that \ref{lambda bounded from below} and \ref{lambda Prime bounded from below} hold.
	\end{generalassumptions}
	
	We switch from the fluid situation to the collisionless situation. But we do not leave the fluid situation completely behind. There is a deep connection between spherically symmetric, fluid and collisionless models, which was elaborated in detail in \cite{2007Rein}. Using this connection, enables us to lift a fluid model $\rho_0=\rho_0(x)$, as constructed in the previous two sections, to a collisionless model, described by a distribution function $f_0=f_0(x,v)$. This distribution function is an equilibrium solution of the collisionless Boltzmann equation \eqref{Boltzmann equation} and solves the variational problem from the introduction. First, we want to recall the variational problem from the introduction:
	
	We take an ansatz function $\Phi$, that satisfies the following assumptions.
	
	\begin{assumptionsPhi}
		$\Phi\in C^1([0,\infty))$, $\Phi(0)=\Phi'(0)=0$ and it holds:
		\begin{enumerate}[label=($\Phi$\arabic*)]
			\item $\Phi$ is strictly convex,
			\item $\Phi(f)\geq Cf^{1+1/k}$ for $f\geq 0$ large, where $0 < k < 3/2$.
		\end{enumerate}
	\end{assumptionsPhi}

	For a distribution function $f\in L^1(\R^6)$, $\geq 0$, we define the Casimir functional
	\begin{equation*}
		\mathcal C (f) := \iint \Phi(f) \diff x \diff v.
	\end{equation*}
	The density $\rho_f$ that belongs to $f$ is given by
	\begin{equation*}
		\rho_f (x) := \int f(x,v) \diff v, \quad x\in\R^3.
	\end{equation*}
	The potential energy of $f$ is the potential energy of its density, i.e.,
	\begin{equation*}
		\Epot^M(f) := \Epot^M(\rho_f) = \Epot^N(\rho_f) + \Epot^Q(\rho_f),
	\end{equation*}
	and its kinetic energy is given by
	\begin{equation*}
		\Ekin(f) := \frac{1}{2} \iint |v|^2 f(x,v) \diff x \diff v.
	\end{equation*}
	We define the functional
	\begin{equation*}
		\Hb (f) := \Epot^M(f) + \Ekin(f) + \mathcal C(f);
	\end{equation*}
	the subscript `B' indicates that we use this functional to construct a stable equilibrium solution of the collisionless Boltzmann equation. We fix a mass $M>0$ and solve the variational problem
	\begin{equation*} 
		\text{minimize } \Hb (f) \text{ s.t. } f \in \mathcal F_M
	\end{equation*}
	where
	\begin{equation*}
		\mathcal F_M := \left\{ f\in L^1(\R^6) \text{ sph. sym} \middle| f \geq 0,\, \|f\|_1=M,\, |\Epot^Q(f)|+\Ekin(f)+\mathcal{C}(f) <\infty \right\}.
	\end{equation*}
	Note that a distribution function $f$ is called spherically symmetric if for all $A\in SO(3)$
	\begin{equation*}
		f(Ax,Av) = f(x,v) \quad \text{for a.e. }x,v\in\R^3.
	\end{equation*}
	
	We want to solve the above variational problem. The basic idea is to reduce the problem, which is defined for distribution functions $f$, to a variational problem that is defined for densities $\rho$. Then we solve the reduced variational problem first and afterwards lift the solution of the reduced problem to a solution of the original problem. We reduce the functional $\Hb$ by factoring out the $v$-dependence. This we do exactly in the same manner as \cite{2007Rein}:
	
	For $r\geq 0$ define
	\begin{equation*}
	\mathcal G_r := \left\{ g\in L^1(\R^3) \text{ sph. sym} \middle| g\geq 0,\,\int g(v)\diff v = r, \, \int\left(\frac 12 |v|^2g(v)+\Phi(g(v))\right)\diff v < \infty \right\}
	\end{equation*}
	and
	\begin{equation*}
	\Psi(r) := \inf_{g\in \mathcal G_r} \int \left( \frac 12 |v|^2g(v)+\Phi(g(v)) \right) \diff v.
	\end{equation*}
	The relation between $\Phi$ and $\Psi$ arises in a natural way. More details on that can be found in \cite{2007Rein}. But the relation between $\Phi$ and $\Psi$ can also be made more explicit using Legendre transformations. This is done in Lemma 2.3. of \cite{2007Rein}. From this relation one deduces the following properties of $\Psi$:
	
	\begin{lem}
		$\Psi \in C^1([0,\infty))$, $\Psi(0)=\Psi'(0)=0$ and it holds:
		\begin{enumerate}[label=($\Psi$\arabic*)]
			\item $\Psi$ is strictly convex,
			\item $\Psi(\rho)\geq C\rho^{1+1/n}$ for $\rho> 0$ large, where $n := k + \frac{3}{2}$.
		\end{enumerate}
	\end{lem}
	
	\begin{proof}
		The proof of this lemma is identical with the proof of Lemma 2.3. in \cite{2007Rein}.
	\end{proof}
	
	In $\Psi$ we have hidden the information about the $v$-dependent terms $\Ekin(f)$ and $\mathcal C(f)$ appearing in the definition of $\Hb$. This way we can reduce the variational problem to minimize $\Hb$ over $\mathcal F_M$ to the following problem:
	\begin{equation*}
		\text{minimize } \He(\rho) = \Epot^M(\rho) + \int \Psi(\rho) \diff x \text{ s.t. } \rho \in \mathcal R_M,
	\end{equation*}
	where $\mathcal R_M$ is as in Section \ref{section existance minimizer pressure supported}. Since $\Psi$ satisfies all assumptions that we used in the Sections \ref{section existance minimizer pressure supported} and \ref{section Euler Lagrange equation pressure supported}, we know that there exists a minimizer $\rho_0\in\mathcal R_M$ of $\He$ over $\mathcal R_M$, which satisfies in particular the Euler-Lagrange equation
	\begin{equation*}
	\rho_0(x) = \begin{cases}
	(\Psi')^{-1}(E_0 -U^M_{\rho_0}(x)), & \text{if } U^M_{\rho_0}<E_0, \\
	0, & \text{if } U^M_{\rho_0}\geq E_0,
	\end{cases}
	\end{equation*}
	for an $E_0\in\R$. Making use of this equation we can lift $\rho_0$ to a minimizer $f_0\in\mathcal F_M$ of the functional $\Hb$. This is the goal of this section and this we do in the following theorem.
	
	\begin{thm} \label{thm f0 is minimizer of HC and definition of f0}
		For every $f\in\mathcal F_M$
		\begin{equation*}
		\Hb(f) \geq \He(\rho_f).
		\end{equation*}
		Let $\rho_0\in\mathcal R_M$ be a minimizer of $\He$ over $\mathcal R_M$ and set
		\begin{equation*}
		f_0(x) := \begin{cases}
		(\Phi')^{-1}(E_0 -E), & \text{if } E<E_0, \\
		0, & \text{if } E\geq E_0,
		\end{cases}
		\end{equation*}
		with $E(x,v)=|v|^2/2 + U^M_{\rho_0}(x)$, $x,v\in\R^3$. Then $\rho_{f_0} = \rho_0$,
		\begin{equation*}
		\Hb(f_0) = \He(\rho_0)
		\end{equation*}
		and $f_0\in\mathcal F_M$ is a minimizer of $\Hb$ over $\mathcal F_M$. This map between minimizers of $\He$ and $\Hb$ is one-to-one and onto. We have further
		\begin{equation*}
		f_0\in C_c(\R^6).
		\end{equation*}
	\end{thm}

	\begin{rem}
		Since $f_0$ is a function of the local energy $E$, $f_0$ is constant along solutions of
		\begin{align*}
			\dot{x} & = v, \\
			\dot{v} & = -\nabla_x U^M_{\rho_0}(t,x). \nonumber
		\end{align*}
		Thus $f_0$ is an equilibrium solution of the collisionless Boltzmann equation \eqref{Boltzmann equation} in the sense of \eqref{characteristic system}.
	\end{rem}
	
	\begin{proof}[Proof of Theorem \ref{thm f0 is minimizer of HC and definition of f0}]
		Inequality (2.11) from \cite{2007Rein} states that for every $f\in\mathcal F_M$
		\begin{equation*}
		\mathcal C(f) + \Ekin(f) \geq \int \Psi(\rho_f) \diff x.
		\end{equation*}
		Thus $\rho_f\in \mathcal R_M$ and 
		\begin{equation} \label{equ proof EL equ f0 HC geq Hr}
		\Hb(f) \geq \He(\rho_f).
		\end{equation}
		With exactly the same proof as for Theorem 2.2. in \cite{2007Rein} if follows that
		\begin{equation} \label{equ proof EL equ f0 HC(f0) = Hr(rho0)}
		\Hb(f_0) = \He(\rho_0)
		\end{equation}
		and
		\begin{equation*}
		\rho_{f_0} = \rho_0;
		\end{equation*}
		one has only to replace $U$ -- which represents the Newtonian potential in \cite{2007Rein} -- by $U^M$ everywhere. \eqref{equ proof EL equ f0 HC geq Hr}, \eqref{equ proof EL equ f0 HC(f0) = Hr(rho0)} and that $\rho_0$ is a minimizer of $\He$ over $\mathcal R_M$ imply that $f_0$ is a minimizer of $\Hb$ over $\mathcal F_M$. Analog to \cite[Lemma 2.3]{2007Rein} this map between minimizers of $\He$ and $\Hb$ is one-to-one and onto with the inverse map
		\begin{equation*}
			\mathcal F_M \ni f_0 \mapsto \rho_0 := \rho_{f_0} \in \mathcal R_M.
		\end{equation*}
		
		By Lemma \ref{lemma minimizers are continuous} $U^M_{\rho_0}\in C^1(\R^3)$. Thus $f_0\in C(\R^6)$. By Lemma \ref{lemma minimizers have compact support}
		\begin{equation*}
		\supp \rho_0 = B_{R_0}
		\end{equation*}
		for an $R_0>0$. Thus there is a $C_0>0$ such that
		\begin{equation*}
		|U^M_{\rho_0}| \leq C_0 \text{ on } \supp \rho_0.
		\end{equation*}
		By the definition of $f_0$, for all $(x,v)\in\supp f_0$ holds 
		\begin{equation*}
		|v|\leq \sqrt{2(E_0-U^M_{\rho_0}(x))} \leq \sqrt{2(|E_0|+C_0)} =: R_1.
		\end{equation*}
		Hence
		\begin{equation*}
		\supp f_0 \subset B_{R_0} \times B_{R_1}
		\end{equation*}
		is compact.		
	\end{proof}

	\section{Stability} \label{section stability}
	
	\begin{generalassumptions}
		Throughout Section \ref{section stability} we assume that \ref{lambda bounded from below} and \ref{lambda Prime bounded from below} hold.
	\end{generalassumptions}
	
	\subsection{Stability in the collisionless situation} \label{section stability collisionless}
	
	To study stability, the first thing we need is a proper notion of time dependent solutions. Reasonable, time dependent solutions $f=f(t,x,v)$ of the collisionless Boltzmann equation \eqref{Boltzmann equation} should have the following properties:
	
	\begin{enumerate}[label=($f$\arabic*)]
		\item \label{reasonable solution f geq 0, sph sym, comp supp} If the initial data $f(0)$ is non-negative, spherically symmetric and compactly supported, then the solution $f(t)$ remains non-negative, spherically symmetric and compactly supported for all times $t>0$.
		\item \label{reasonable solution f preserves Lp norm} The solution preserves $L^p$-norms, i.e., for every $1\leq p \leq \infty$ and $t>0$
		\begin{equation*}
		\| f(t) \|_p = \| f(0) \|_p.
		\end{equation*}
		\item \label{reasonable solution f preserves energy} The solution preserves energy, i.e., for all $t>0$
		\begin{equation*}
		\Epot^M(f(t)) + \Ekin(f(t)) = \Epot^M(f(0)) + \Ekin(f(0)).
		\end{equation*}
		\item \label{reasonable solution f preserves Casimir functional} The flow corresponding to the collisionless Boltzmann equation conserves phase space volume and as a consequence for all measurable functions $\Phi:[0,\infty) \rightarrow [0,\infty)$ the following conservation law holds:
		\begin{equation*}
		\iint \Phi( f(t) ) \diff x \diff v = \iint \Phi ( f(0) ) \diff x \diff v
		\end{equation*}
	\end{enumerate}
	\begin{defn}
		We call a distribution function $f=f(t,x,v)$ that solves the collisionless Boltzmann equation \eqref{Boltzmann equation} in some sense a \textit{reasonable} solution if it satisfies \ref{reasonable solution f geq 0, sph sym, comp supp} to \ref{reasonable solution f preserves Casimir functional}.
	\end{defn}
	
	In her master thesis, Carina Keller has recently proven that such reasonable solutions of the collisionless Boltzmann equation coupled with the Mondian field equations exist. Given that the initial data is spherically symmetric and reasonably well behaved, she has also proven that these solutions exist for all times $t>0$. This result is under preparation for being published.
	
	The next thing we need, is a suitable tool to measure the distance between $f_0$ and another distribution function $f\in\mathcal F_M$. A first order Taylor expansion of $\Epot^M$ under the integral sign gives for $f\in\mathcal F_M$
	\begin{equation*}
	\Hb(f) - \Hb(f_0) = d(f,f_0) + \text{remainder terms}
	\end{equation*}
	where
	\begin{equation*}
	d(f,f_0) := \iint \left( \Phi(f) - \Phi(f_0) + E(f-f_0) \right) \diff x \diff v;
	\end{equation*}
	as before $E=|v|^2/2 + U^M_{f_0}(x)$, $x,v\in\R^3$. In the sequential, we use the notation
	\begin{equation*}
		 U_0^M :=  U_{f_0}^M.
	\end{equation*}
	For $d$ the following statements holds:
	
	\begin{lem} \label{lemma d is approriate to measure distances}
		For every $f\in\mathcal F_M$ $d(f,f_0)\geq 0$ and $d(f,f_0)=0$ if and only if $f=f_0$.
	\end{lem}
	
	\begin{proof}
		Since $\iint f\diff x\diff v = \iint f_0\diff x\diff v = M$ and $\Phi$ is convex,
		\begin{align*}
		d(f,f_0) & = \iint \left( \Phi(f) - \Phi(f_0) + (E-E_0)(f-f_0) \right) \diff x\diff v\\
		& \geq \iint \left[ \Phi'(f_0) + E - E_0 \right] (f-f_0) \diff x \diff v.
		\end{align*}
		Due to the definition of $f_0$ the term in the brackets vanishes if $f_0>0$. Hence
		\begin{equation*}
		d(f,f_0) \geq 0.
		\end{equation*}
		Moreover, since $\Phi$ is strictly convex, there is for every $f\in\mathcal F_M$ with $f\neq f_0$ a set of positive measure where
		\begin{equation*}
		\Phi(f)-\Phi(f_0) > \Phi'(f_0)(f-f_0).
		\end{equation*}
		Hence $d(f,f_0) = 0$ if and only if $f=f_0$.
	\end{proof}
	
	Thus $d$ is an appropriate tool to measure distances between $f_0$ and $f\in\mathcal F_M$. We have to take care of the remainder terms. Since
	\begin{equation*}
	\Epot^N(f) = -\frac{1}{8\pi} \int |\nabla U^N_f|^2 \diff x = - \frac{1}{2} \iint \frac{\rho_f(x)\rho_f(y)}{|x-y|}\diff x \diff y
	\end{equation*}
	is quadratic in $f$, its Taylor expansion stops after the second term and the corresponding remainder term is simple. For $\Epot^Q$ the Taylor expansion does not stop and we have to estimate the corresponding remainder term using Lemma \ref{lemma Q differentiable}. This leads to the following lemma.
	
	\begin{lem} \label{lemma Taylor of HC}
		Let $f\in \mathcal F_M\cap L^\infty(\R^6)$ be compactly supported, then
		\begin{equation*}
		\left| \Hb(f) - \Hb(f_0) - d(f,f_0) \right| \leq C\left( \|\nabla U^N_f -\nabla U^N_0 \|_2^2 + \|\nabla U^N_f - \nabla U^N_0\|_{3/2}^{3/2} \right).
		\end{equation*}
	\end{lem}
	
	\begin{proof}
		We have
		\begin{equation} \label{equ proof Taylor of HC complete expansion}
		\Hb(f) - \Hb(f_0) - d(f,f_0) = \Epot^M(f) - \Epot^M(f_0) - \iint U^M_0(f-f_0) \diff x \diff v.
		\end{equation}
		A second order Taylor expansion under the integral sign gives
		\begin{align*}
		\Epot^N(f)  &- \Epot^N(f_0)= - \frac{1}{8\pi} \int \left(|\nabla U^N_f|^2 - |\nabla U_0|^2 \right)\diff x \\
		& = -\frac{1}{4\pi} \int\nabla U^N_0\cdot(\nabla U^N_f-\nabla U^N_0)\diff x - \frac{1}{8\pi}\int |\nabla U^N_f - \nabla U^N_0|^2 \diff x.
		\end{align*}
		Since $\rho_f,\rho_0\in\mathcal R_M\subset L^{6/5}(\R^3)$, Lemma \ref{lemma EpotN for rho in L6/5} gives
		\begin{equation*}
		-\frac{1}{4\pi} \int \nabla U^N_0 \cdot (\nabla U^N_f - \nabla U^N_0) \diff x = \iint U^N_0(f-f_0) \diff x \diff v.
		\end{equation*}
		Thus we have
		\begin{equation} \label{equ proof Taylor of HC formula for EpotN}
		\Epot^N(f) - \Epot^N(f_0) - \iint U^N_0(f-f_0) \diff x\diff v = -\frac{1}{8\pi} \|\nabla U^N_f - \nabla U^N_0\|_2^2.
		\end{equation}
		Further Lemma \ref{lemma Q differentiable} implies
		\begin{align*}
		\left| \Epot^Q(f) - \Epot^Q(f_0) + \frac{1}{4\pi}\int \lambda(|\nabla U^N_0|)\nabla U^N_0\cdot (\nabla U^N_f - \nabla U^N_0) \diff x \right| \leq C\|\nabla U^N_f - \nabla U^N_0\|_{3/2}^{3/2}.
		\end{align*}
		Since $\supp \rho_f, \supp \rho_0$ are compact, $\|\rho_f\|_\infty,\|\rho_0\|_\infty<\infty$ and $\int \rho_f\diff x = \int \rho_0\diff x =M$, we get as in the proof of Lemma \ref{lemma derivative of EpotQ}
		\begin{equation*}
		\frac{1}{4\pi}\int\lambda(|\nabla U^N_0|)\nabla U^N_0\cdot (\nabla U^N_f - \nabla U^N_0) \diff x = -\iint U^\lambda_0(f-f_0) \diff x \diff v.
		\end{equation*}
		Thus
		\begin{equation} \label{equ proof Taylor of HC formula for EpotQ}
		\left|\Epot^Q(f) - \Epot^Q(f_0) - \iint U^\lambda_0(f-f_0) \diff x\diff v \right|\leq C \|\nabla U^N_f - \nabla U^N_0\|_{3/2}^{3/2}.
		\end{equation}
		Taking \eqref{equ proof Taylor of HC complete expansion}, \eqref{equ proof Taylor of HC formula for EpotN} and \eqref{equ proof Taylor of HC formula for EpotQ} together implies
		\begin{equation*}
		\left| \Hb(f) - \Hb(f_0) - d(f,f_0) \right| \leq C\left( \|\nabla U^N_f -\nabla U^N_0 \|_2^2 + \|\nabla U^N_f - \nabla U^N_0\|_{3/2}^{3/2} \right).
		\end{equation*}		
	\end{proof}	
	
	Now we have everything we need to prove the following stability result.
	
	\begin{thm} \label{thm stability}
		Assume that the minimizer $f_0\in\mathcal F_M$ of $\Hb$ over $\mathcal F_M$ from Theorem \ref{thm f0 is minimizer of HC and definition of f0} is unique. Then for every $\epsilon>0$ there is a $\delta>0$ such that for every reasonable solution $f$ of the collisionless Boltzmann equation with $f(0)\in \mathcal F_M\cap L^\infty(\R^6)$, $f(0)$ compactly supported and 
		\begin{equation*}
		d(f(0),f_0) + \|\nabla U^N_{f(0)} - \nabla U^N_0\|_2 + \| \nabla U^N_{f(0)} - \nabla U^N_0\|_{3/2} < \delta
		\end{equation*}
		holds
		\begin{equation*}
		d(f(t),f_0) + \|\nabla U^N_{f(t)} - \nabla U^N_0\|_2 + \| \nabla U^N_{f(t)} - \nabla U^N_0\|_{3/2} < \epsilon, \quad 0\leq t < \infty.
		\end{equation*}
		
	\end{thm}
	
	\begin{proof}
		Assume that there exist reasonable solutions $f_j$ of the collisionless Boltzmann equation with $f_j(0)\in\mathcal F_M\cap L^\infty(\R^6)$, $(t_j)\subset [0,\infty)$ and $\epsilon>0$
		such that
		\begin{equation} \label{equ proof stability norm initially < 1/j}
		d(f_j(0),f_0) + \|\nabla U^N_{f_j(0)} - \nabla U^N_0\|_2 + \| \nabla U^N_{f_j(0)} - \nabla U^N_0\|_{3/2} < \frac{1}{j}
		\end{equation}
		but
		\begin{equation} \label{equ proof stability norm at time tj > epsilon}
		d(f_j(t_j),f_0) + \|\nabla U^N_{f_j(t_j)} - \nabla U^N_0\|_2 + \| \nabla U^N_{f_j(t_j)} - \nabla U^N_0\|_{3/2} > \epsilon.
		\end{equation}
		First we check that $f_j(t_j)\in\mathcal F_M$ for every $j\in\N$. From \ref{reasonable solution f geq 0, sph sym, comp supp} we know that $f_j(t_j)$ is non-negative, spherically symmetric and has compact support. Further \ref{reasonable solution f preserves Lp norm} ensures that $f_j$ preserves $L^p$-norms. Thus
		\begin{equation*}
		\iint f_j(t_j) \diff x \diff v = \|f_j(t_j)\|_1 = \|f_j(0)\|_1 = M
		\end{equation*}
		and
		\begin{equation*}
		\|f_j(t_j)\|_\infty = \|f_j(0)\|_\infty < \infty.
		\end{equation*}
		Since $f_j(t_j)$ is compactly supported and bounded,
		\begin{equation*}
		|\Epot^Q(f_j(t_j))| + \Ekin(f_j(t_j)) + \mathcal C(f_j(t_j)) < \infty.
		\end{equation*}
		Thus $f_j(t_j)\in \mathcal F_M$.
		
		Now \eqref{equ proof stability norm initially < 1/j} and Lemma \ref{lemma Taylor of HC} imply that
		\begin{equation*}
		\Hb(f_j(0)) \rightarrow \Hb(f_0) \quad \text{for }j\rightarrow \infty.
		\end{equation*}
		By \ref{reasonable solution f preserves energy} and \ref{reasonable solution f preserves Casimir functional} the total energy and the Casimir functional are conserved quantities of $f_j$. Thus
		\begin{equation*}
		\Epot^M(f_j(t_j)) + \Ekin(f_j(t_j)) = \Epot^M(f_j(0)) + \Ekin(f_j(0)).
		\end{equation*}
		and
		\begin{align*}
		\mathcal C (f_j(t_j)) = \iint \Phi(f_j(t_j)) \diff x \diff v = \iint \Phi(f_j(0)) \diff x \diff v 
		= \mathcal{C}(f_j(0)).
		\end{align*}
		Hence
		\begin{equation} \label{equ proof stability convergence of H(fj(tj))}
		\Hb(f_j(t_j)) = \Hb(f_j(0)) \rightarrow \Hb(f_0) \quad \text{for }j\rightarrow\infty.
		\end{equation}
		Thus $(f_j(t_j))\subset \mathcal F_M$ is a minimizing sequence of $\Hb$ over $\mathcal{F}_M$. Theorem \ref{thm f0 is minimizer of HC and definition of f0} implies that $(\rho_{f_j}(t_j))\subset\mathcal R_M$ is also a minimizing sequence of $\He$ over $\mathcal R_M$. Hence Theorem \ref{thm existence of minimizers of Hr} implies
		\begin{equation*}
		\|\nabla U^N_{f_j(t_j)} - \nabla U^N_0 \|_2 + \|\nabla U^N_{f_j(t_j)} - \nabla U^N_0 \|_{3/2} \rightarrow 0 \quad \text{for }j\rightarrow\infty.
		\end{equation*}
		Together with \eqref{equ proof stability convergence of H(fj(tj))} and Lemma \ref{lemma Taylor of HC} this implies
		\begin{equation*}
		d(f_j(t_j),f_0) \rightarrow 0 \quad \text{for }j\rightarrow\infty.
		\end{equation*}
		Thus for $j$ sufficiently large
		\begin{equation*}
		d(f_j(t_j),f_0) + \|\nabla U^N_{f_j(t_j)} - \nabla U^N_0\|_2 + \| \nabla U^N_{f_j(t_j)} - \nabla U^N_0\|_{3/2} < \epsilon,
		\end{equation*}
		which contradicts the assumption \eqref{equ proof stability norm at time tj > epsilon}.
		
	\end{proof}
	
	\subsection{Stability in the fluid situation}
	
	We turn from the collisonless Boltzmann equation \eqref{Boltzmann equation} to the Euler equations \eqref{Euler equations}. With an analog proof as in Section \ref{section stability collisionless} we prove that $\rho_0$ from Section \ref{section existance minimizer pressure supported} and \ref{section Euler Lagrange equation pressure supported} is stable against small spherically symmetric perturbations. First we need a concept of \textit{reasonable} time dependent solutions of the Euler equations.
	
	\begin{defn}
		Let $T>0$, $\rho:[0,T)\times \R^3 \rightarrow [0,\infty)$ be a density and $u:[0,T)\times \R^3 \rightarrow \R^3$ be a velocity field such that the tuple $(\rho,u)$ solves the Euler equations \eqref{Euler equations} in some sense on the time interval $[0,T)$. We call the tuple $(\rho,u)$ a \textit{reasonable} solution of the Euler equations if \ref{reasonable solution rho geq 0, sph sym} to \ref{reasonable solution rho preserves energy} hold.
		\begin{enumerate}[label=($\rho$\arabic*)]
			
			\item \label{reasonable solution rho geq 0, sph sym} If the initial data $\rho(0)$ is non-negative, spherically symmetric, compactly supported and bounded, then the solution $\rho(t)$ remains non-negative, spherically symmetric, compactly supported and bounded for all $t\in[0,T)$.
			
			\item \label{reasonable solution rho preserves mass} The solution preserves mass, i.e., for all $t\in[0,T)$
			\begin{equation*}
			\| \rho(t) \|_1 = \| \rho(0) \|_1.
			\end{equation*}
			\item \label{reasonable solution rho preserves energy} The solution preserves energy, i.e., for all $t\in[0,T)$
			\begin{equation*}
				E(\rho(t),u(t)) = E(\rho(0),u(0))
			\end{equation*}
			where
			\begin{equation*}
				E(\rho(t),u(t)) := \Epot^M(\rho(t)) + \frac{1}{2} \int |u(t)|^2\rho(t) \diff x + \int \Psi(\rho(t)) \diff x.
			\end{equation*}
		\end{enumerate}
	\end{defn}

	In contrast to the notion of reasonable solutions of the collisionless Boltzmann equation, we define reasonable solution of the Euler equations only on some (possibly finite) time interval $[0,T)$. The reason for this is that it is unknow under which conditions global in time solutions of the Euler equations exists. This is even unknown when the Euler equations are coupled with the simpler Newtonian field equations. Further, it is also not know under which conditions solutions to the initial value problem of the Euler equation really preserve energy; neither in Newtonian nor in Mondian physics. Thus the stability result for fluid models is somewhat weaker than its collisionless counterpart. However, we stress that this weakness only arises from an incomplete understanding of time dependent solutions of the Euler equations. Our stability analysis itself is perfectly rigorous.
	
	To prove stability, we need tools to measure distances and we construct them analogously to the collisionless situation. Let $\rho\in\mathcal R_M$ and let $u$ be a velocity field such that the kinetic energy of $(\rho,u)$ is finite. Expanding $\Epot^M$ as before, we find
	\begin{equation*}
		E(\rho,u) - E(\rho,0) = \frac{1}{2} \int |u|^2\rho\diff x + d(\rho,\rho_0) + \text{remainder terms},
	\end{equation*}
	where now
	\begin{equation*}
		d(\rho,\rho_0) := \int \left[ \Psi(\rho) - \Psi(\rho_0) + U^M_0(\rho-\rho_0) \right];
	\end{equation*}
	we use the notion $U^M_0 := U^M_{\rho_0}$. As in Lemma \ref{lemma d is approriate to measure distances} we get for $\rho\in\mathcal R_M$ that $d(\rho,\rho_0)\geq 0$ and $d(\rho,\rho_0) = 0$ if and only if $\rho=\rho_0$. Since
	\begin{equation*}
		E(\rho,u) - E(\rho,0) - \frac{1}{2} \int |u|^2\rho\diff x - d(\rho,\rho_0) = \Epot^M(\rho) - \Epot^M(\rho_0) - \int U^M_0(\rho-\rho_0) \diff x,
	\end{equation*}
	we can estimate the remainder terms exactly in the same way as done in Lemma \ref{lemma Taylor of HC}. Thus with the same proof as in the collisionless situation, the following stability result for $\rho_0$ follows:
	
	\begin{thm} \label{thm stability Euler}
		Assume that the minimizer $\rho_0\in\mathcal R_M$ of $\He$ over $\mathcal R_M$ from Theorem \ref{thm existence of minimizers of Hr} is unique. Then for every $\epsilon>0$ there is a $\delta>0$ such that for every reasonable solution $(\rho,u)$ of the Euler equations with $\rho(0)\in \mathcal R_M\cap L^\infty(\R^3)$, $\rho(0)$ compactly supported and 
		\begin{equation*}
		\frac{1}{2} \int |u(0)|^2\rho(0) \diff x + d(\rho(0),\rho_0) + \|\nabla U^N_{\rho(0)} - \nabla U^N_0\|_2 + \| \nabla U^N_{\rho(0)} - \nabla U^N_0\|_{3/2} < \delta
		\end{equation*}
		holds
		\begin{equation*}
		\frac{1}{2} \int |u(t)|^2\rho(t) \diff x + d(\rho(t),\rho_0) + \|\nabla U^N_{\rho(t)} - \nabla U^N_0\|_2 + \| \nabla U^N_{\rho(t)} - \nabla U^N_0\|_{3/2} < \epsilon, \quad 0\leq t < T.
		\end{equation*}
		
	\end{thm}

	\section{Discussion about spherical symmetry} \label{section getting rid of spherical symmetry}
	
	Is it possible to prove the stability results from Section \ref{section stability} without symmetry assumptions? The first point where we used the assumption of spherical symmetry was Lemma \ref{lemma estimates for Epot} where  we proved bounds for the potential energy. For the Mondian part of the potential energy this proof relied very much on the assumption of spherical symmetry. In order to get a similar bound also without symmetry assumption we suspect that one has to prove first a confining property like the one from Lemma \ref{lemma masses remain concentrated along minimizing sequences} where $\Epot^Q(\rho)$ controls how far apart the mass of $\rho$ can be scattered. But also this proof used spherical symmetry and a new idea is necessary to prove such a confining property without symmetry assumptions. We suspect that this could be achieved using the following, \textit{formal} property for the Mondian part of the potential energy. For simplicity we set $\lambda(u)=1/\sqrt u$, $u>0$, and thus $Q(u)=2/3 \, u^{3/2}$.  Then
	
	\begin{align*}
		- \frac{1}{4\pi} \int Q\left(\left| \nabla U^N_\rho \right|\right) \diff x
		&  = - \frac{1}{6\pi} \int \left| \nabla U^N_\rho \right|^{3/2} \diff x
		= - \frac{1}{6\pi} \int \nabla U^N_\rho \cdot \frac{\nabla U^N_\rho}{\left| \nabla U^N_\rho \right|^{1/2}} \diff x
	\end{align*}
	As already mentioned in Section \ref{section Mondian potentials} $\nabla U^\lambda_\rho$ is the irrotational part of $\nabla U^N_\rho / \left| \nabla U^N_\rho \right|^{1/2}$ and we write
	\begin{equation*}
	\nabla U^\lambda_\rho = H\left(\frac{\nabla U^N_\rho}{\left| \nabla U^N_\rho \right|^{1/2}}\right);
	\end{equation*}
	where $H$ is the operator that extracts the irrotational part of a vector field \citep[Definition 3.1.]{2024Frenkler}. Using that every gradient is already irrotational and that $H$ is symmetric \citep[Lemma 3.6.]{2024Frenkler} we get
	\begin{align*}
		- \frac{1}{4\pi} \int Q\left(\left| \nabla U^N_\rho \right|\right) \diff x
		& = - \frac{1}{6\pi} \int H\left(\nabla U^N_\rho\right) \cdot \frac{\nabla U^N_\rho}{\left| \nabla U^N_\rho \right|^{1/2}} \diff x
		= - \frac{1}{6\pi} \int \nabla U^N_\rho \cdot H\left(\frac{\nabla U^N_\rho}{\left| \nabla U^N_\rho \right|^{1/2}}\right) \diff x\\
		& = - \frac{1}{6\pi} \int \nabla U^N_\rho \cdot \nabla U^\lambda _\rho \diff x
	\end{align*}
	With integration by parts
	\begin{align*}
		- \frac{1}{4\pi} \int Q\left(\left| \nabla U^N_\rho \right|\right) \diff x & = \frac{1}{6\pi} \int \Delta U^N_\rho \cdot U^\lambda_\rho \diff x + \text{border terms} = \frac{2}{3} \int \rho \, U^\lambda_\rho \diff x  + \text{border terms}.
	\end{align*}
	Such a relation would allow us to relate the density $\rho$ to the Mondian part $U^\lambda_\rho$ of the potential $U^M_\rho$. Since $U^\lambda_\rho$ should diverge logarithmically at infinity, this can be used to prove the desired confining property. But one has to treat the above, formal relation with the utmost care! First, $\int Q (\ldots)$ is not finite. One has always to study energy differences in MOND. At the end of the above derivation, the infinite terms are hidden in the border terms. And second, $U^\lambda_\rho$ is only defined up to an additive constant. Since $U^\lambda_\rho$ diverges at infinity, it is unclear how it should be normalized. The correct normalization constant will depend on the reference density $\bar \rho$ that must be introduced to get a finite value for the integral $\int [Q(\ldots) - Q(\ldots) ]$. Thus there are several arguments that must be elaborated in detail.
	
	The next problem that one faces is that for a minimizing sequence $(\rho_j)$ that converges weakly to a minimizer $\rho_0$ we must prove
	\begin{equation*}
	\nabla U^N_{\rho_j} - \nabla U^N_{\rho_0} \rightarrow 0 \quad \text{strongly in } L^{3/2}(\R^3).
	\end{equation*}
	This is an essential ingredient in the proof of the stability statements in the Theorems \ref{thm stability} and \ref{thm stability Euler}. Proofing that $\nabla  U^N_{\rho_j} - \nabla U^N_{\rho_0}$ converges strongly to zero in $L^{3/2}(\R^3)$ relied on the assumption of spherical symmetry, too, and it did so in a somewhat sophisticated way. Also there we need a new proof. Probably the best ansatz is to search for a generalization of Lemma \ref{lemma compactness of Laplace inverse}, which is taken from Lemma 2.5 in \cite{2007Rein}. There it was proven that if a weakly convergent sequence of densities remains concentrated, then the $L^2$-norm of the corresponding Newtonian potentials converges. This must be generalized to a convergence in the $L^{3/2}$ norm.
	
	If we manage to prove that without the assumption of spherical symmetry there are minimizers for the variational problem from Section \ref{section existance minimizer pressure supported}, it is easy to get the corresponding Euler-Lagrange equation since our proof from Section \ref{section Euler Lagrange equation pressure supported} did not depend on the assumption of spherical symmetry. Further we would expect that the minimizers are still spherically symmetric.
	This holds in the Newtonian situation of the Vlasov-Poisson system and there is no reason why this should be different in the Mondian situation of the \VQMS. In the Newtonian situation we are aware of two proofs to show that minimizers are still spherically symmetric but both proofs fail in the Mondian situation:
	
	The first proof uses that under the symmetric rearrangement of a density $\rho$ the corresponding potential energy $\Epot^N(\rho)$ decreases \cite[Theorem 3.7]{2010LiebLoss}. But this argument relies on the fact that the potential energy can be written in the form
	\begin{equation*}
	\Epot^N(\rho) = - \frac 12 \iint \frac{\rho(x)\rho(y)}{|x-y|} \diff x \diff y.
	\end{equation*}
	We do not have such a representation for $\Epot^Q(\rho)$ and therefore this proof cannot be applied in the Mondian situation.
	
	The second proof relies on a symmetry result for solutions of elliptic equations \citep{1979CMaPh..68..209GidasNirenberg}. In the Newtonian situation we can deduce from the Poisson equation and from the Euler-Lagrange equation that the potential $U^N$ of a minimizer $\rho_0$ solves an equation of the form
	\begin{equation*}
	\Delta U^N = 4\pi\rho_0 = 4\pi g(U^N)
	\end{equation*}
	with a suitable function $g$. For such an equation one can deduce from Theorem 4 of \cite{1979CMaPh..68..209GidasNirenberg} that $U^N$ must be spherically symmetric. Now we consider again the Mondian situation and we switch from the QUMOND formulation, which we used in Section \ref{section Potential theory}, to the AQUAL\footnote{The theory is called AQUAL because the field equation derives from an aquadratic Lagrangian \citep{2012LRR....15...10FamaeyMcGaugh}.} formulation of the MOND theory. This enables us to see more directly why the result of \citeauthor{1979CMaPh..68..209GidasNirenberg} cannot be applied in the Mondian situation. In AQUAL the Mondian potential $U^M$ is derived as the solution of the PDE
	\begin{equation*}
	\divergence\left( \mu(|\nabla U^M|)\nabla U^M \right) = 4\pi\rho_0, \quad \lim_{|x|\rightarrow \infty} |\nabla U^M(x)| = 0;
	\end{equation*}
	$\mu$ shall be such that $\mu(\tau)\approx 1$ for $\tau \gg 1$ and $\mu(\tau) \approx \tau$ for $0<\tau \ll 1$. If $\mu$ and $\lambda$ are connected in the correct way, the potentials $U^M$ derived from the AQUAL and the QUMOND theory are identical in spherically symmetric situations. But without the assumption of spherical symmetry these potentials are in general different. As in the Newtonian situation a minimizer $\rho_0$ is connected to the potential $U^M$ via an equation of the form $\rho_0 = g(U^M)$. Then we have
	\begin{align*}
	0 = \divergence & \left( \mu(|\nabla U^M|)\nabla U^M \right) - 4\pi g(U^M) \\
	& = \mu(|\nabla U^M|)\sum_{i=1}^3 \partial_{x_i}^2 U^M + \frac{\mu'(|\nabla U^M|)}{|\nabla U^M|} \sum_{i,j=1}^3 \partial_{x_i} U^M \partial_{x_j} U^M \partial_{x_i}\partial_{x_j} U^M - 4\pi g(U^M) \\
	& = G( U^M, \partial_{x_i} U^M, \partial_{x_i}\partial_{x_j} U^M)
	\end{align*}
	with a suitable function $G=G(a,x,H)$, $a\in\R$, $x\in\R^3$, $H=(h_{ij})\in \R^3\times\R^3$. To apply Theorem 4 of \cite{1979CMaPh..68..209GidasNirenberg} we need that for all $a\in\R$, $x\in\R^3$ the matrix
	\begin{equation*}
	(\partial_{h_{ij}} G)_{1\leq i,j\leq 3} = \mu(|x|)E_3 + \frac{\mu'(|x|)}{|x|}xx^T
	\end{equation*}
	is positive definite; $E_3$ denotes the identity matrix with dimension 3. Assume for simplicity that $\mu(|x|)=|x|$ if $|x|$ is small. Then we have for $|x|$ small
	\begin{equation*}
	(\partial_{h_{ij}} G)_{1\leq i,j\leq 3} = |x|\left(E_3 + \frac{xx^T}{|x|^2} \right).
	\end{equation*}
	Thus
	\begin{equation*}
	(\partial_{h_{ij}} G)_{1\leq i,j\leq 3} \rightarrow 0 \quad \text{for }|x|\rightarrow 0.
	\end{equation*}
	Hence $(\partial_{h_{ij}} G)$ is not positive definite for $x=0$ and the symmetry result of \cite{1979CMaPh..68..209GidasNirenberg} cannot be applied in the Mondian situation.
	
	Summarizing we can say, that with some effort it seems possible to remove the assumption of spherical symmetry from the treatment of the above variational problems. But especially establishing analytically that the resulting model is again spherically symmetric will be quite a challenging task since the classical results from the literature cannot be applied to the Mondian situation directly.
	
	\section{Applying our stability result and verifying the uniqueness condition} \label{section minimizer must be unique}
	
	In this paper we have searched for minimizers of two variational problems and have proved that such minimizers exist  (Theorem \ref{thm existence of minimizers of Hr} and Theorem \ref{thm f0 is minimizer of HC and definition of f0}). Subsequently, we have proved the non-linear stability of these minimizers under the assumption that there is exactly one minimizer of the variational problem (Theorem \ref{thm stability} and Theorem \ref{thm stability Euler}).
	How to deal with this uniqueness requirement?
	In most situations, we are interested in one particular model. Then verifying the uniqueness condition can be done numerically as we demonstrate in the sequential.
	
	Consider a polytropic fluid model. Such a model has an equation of state of the form
	\begin{equation*}
		p(x) = C\rho(x)^\gamma
	\end{equation*}
	with $C,\gamma>0$. We choose our ansatz function $\Psi(\rho) = \frac{1}{2} \rho^2$; this $\Psi$ satisfies all assumptions from the above sections. With \eqref{equ equation of state in introduction} from the introduction, this ansatz leads to the equation of state
	\begin{equation*}
		p(x) = \rho(x)^2 - \frac{1}{2} \rho(x)^2 = \frac{1}{2} \rho(x)^2.
	\end{equation*}
	Let $\rho_0$ be a minimizer of the variational problem from Section \ref{section existance minimizer pressure supported} with mass $M=\|\rho_0\|_1 >0$. We want to show that $\rho_0$ is unique. The Euler-Lagrange equation from Theorem \ref{thm Euler Lagrange equation} tells us that everywhere where $\rho_0$ is greater than zero
	\begin{equation*}
		\rho_0(r) = E_0 - U^M_{\rho_0}(r)
	\end{equation*}
	since $(\Psi')^{-1}(\eta) = \eta$, $\eta \geq 0$. In view of Lemma \ref{lemma Ulambda sph sym}, taking the derivative with respect to $r$ gives
	\begin{align} \label{equ ode for rho}
		\rho_0'(r) & =  - U^{N\prime}_{\rho_0}(r) - \lambda\left( \left| U^{N\prime}_{\rho_0}(r)\right|\right) U^{N\prime}_{\rho_0}(r)  = - \frac{M_{\rho_0}(r)}{r^2} - \frac{\sqrt{M_{\rho_0}(r)}}{r},
	\end{align}
	where we assumed for simplicity that $\lambda(u) = 1/\sqrt{u}$, $u>0$. Let us search all possible solutions of this integro-differential equation. Since $\rho_0$ is continuous (Lemma \ref{lemma minimizers are continuous}), the central value $\rho_0(0)$ is well defined. \cite{2015Rein} has proven that for every $s>0$, there exists exactly one solution of \eqref{equ ode for rho} with central value $s$. Let us call this solution $\rho_s$. Thus $\rho_s(0) = s$, $\rho_s$ solves \eqref{equ ode for rho} on some interval $[0,R_s)$ and is positive there, and $\rho_s$ vanishes on the interval $[R_s,\infty)$.
	
	\begin{figure}
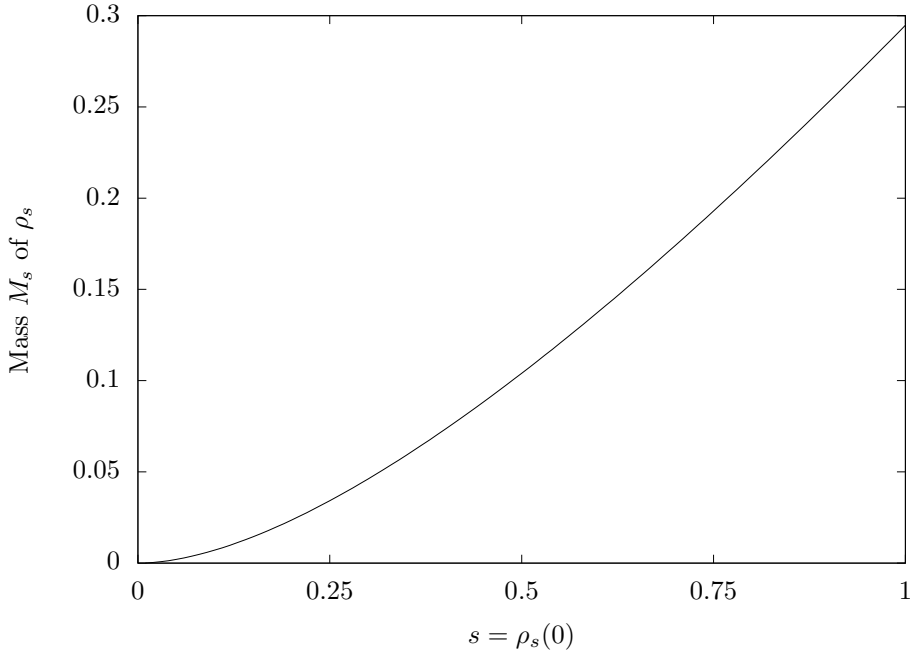

		\begin{center}
			\include{figure-mass-of-minimizer}
		\end{center}
		\caption{We calculated numerically the solution $\rho_s$ of \eqref{equ ode for rho} with central value $\rho_s(0)=s$. In the above figure we plotted the mass $M_s$ of $\rho_s$ against the central value $s\in[0.01,1]$. We see that in the depicted range for every given mass $M>0$ there is only one $s$ such that $\rho_s$ has mass $M$. The behaviour of the graph of $M_s$ for $s$ very small (the deep MOND limit) and for $s$ very large (the Newtonian limit) is given in Figure \ref{figure-mass-of-minimizer-limits}.}
		\label{figure-mass-of-minimizer}
	\end{figure}

	To check whether our minimizer is unique, we wrote a script that calculates for every parameter $s$ the density $\rho_s$. We deliver this script along with this paper. In Figure \ref{figure-mass-of-minimizer} we have plotted the mass $M_s$ of $\rho_s$ against the parameter $s\in[0.01,1]$. We see that for for $s \rightarrow 0$ the graph looks like a parabola and for $s\rightarrow 1$ it becomes more and more a straight line. This is due to the behaviour of $\rho_s$ in the deep MOND limit and the Newtonian limit. If $s$ is small, then the square root term in \eqref{equ ode for rho} dominates and \eqref{equ ode for rho} becomes essentially
	\begin{equation*}
		\rho_s'(r) = -\frac{\sqrt{M_{\rho_s}(r)}}{r}.
	\end{equation*}
	This is the deep MOND limit. From this equation we can derive the scaling relation
	\begin{equation*}
		M_s = \frac{s^2}{s_0^2} M_{s_0}, \quad s,s_0> 0 \text{ small.}
	\end{equation*}
	Thus $M_s$ as a function of $s$ becomes a parabola. This we have visualized in the left diagram of Figure \ref{figure-mass-of-minimizer-limits} where we have zoomed in on the range $s\in[0.0001,0.01]$ and plotted additionally the parabola $1.06\,s^2$. In contrast, in the Newtonian limit, when $s\rightarrow\infty$, equation \ref{equ ode for rho} becomes
	\begin{equation*}
		\rho_s'(r) = - \frac{M_{\rho_s}(r)}{r^2}
	\end{equation*}
	and this leads to the scaling relation
	\begin{equation*}
		M_s = \frac{s}{s_0} M_{s_0}, \quad s,s_0> 0 \text{ large.}
	\end{equation*}
	Thus in the Newtonian limit, $M_s$ becomes a linear function. This we have visualized in the right diagram of Figure \ref{figure-mass-of-minimizer-limits} where we have zoomed out on the range $s\in[10,1000]$.
	
	From the above three plots and the scaling relations in the deep MOND and the Newtonian limit we see that for every prescribed mass $M>0$, there is exactly one $s_0\in[0,\infty)$ such that $\rho_{s_0}$ has mass $M$. Since the family $\{\rho_s\}$ contains all possible solutions of the Euler-Lagrange equation \eqref{equ ode for rho}, this implies that the minimizer $\rho_0=\rho_{s_0}$ of $\He$ over $\mathcal R_M$ is unique. Thus Theorem \ref{thm stability Euler} implies that $\rho_0$ is a non-linearly stable equilibrium solution of the Euler equations \ref{Euler equations}.
	
	Additionally, Theorem \ref{thm f0 is minimizer of HC and definition of f0} tells us that we can lift $\rho_0$ to a minimizer $f_0$ of the variational problem treated in Section \ref{section existance of minimizer in collisionless situation}. Since the lifting process is one-to-one and onto, $f_0$ is unique too. Thus Theorem \ref{thm stability} implies that $f_0$ is a non-linearly stable equilibrium solution of the collisionless Boltzmann equation \eqref{Boltzmann equation}.

	\begin{figure}
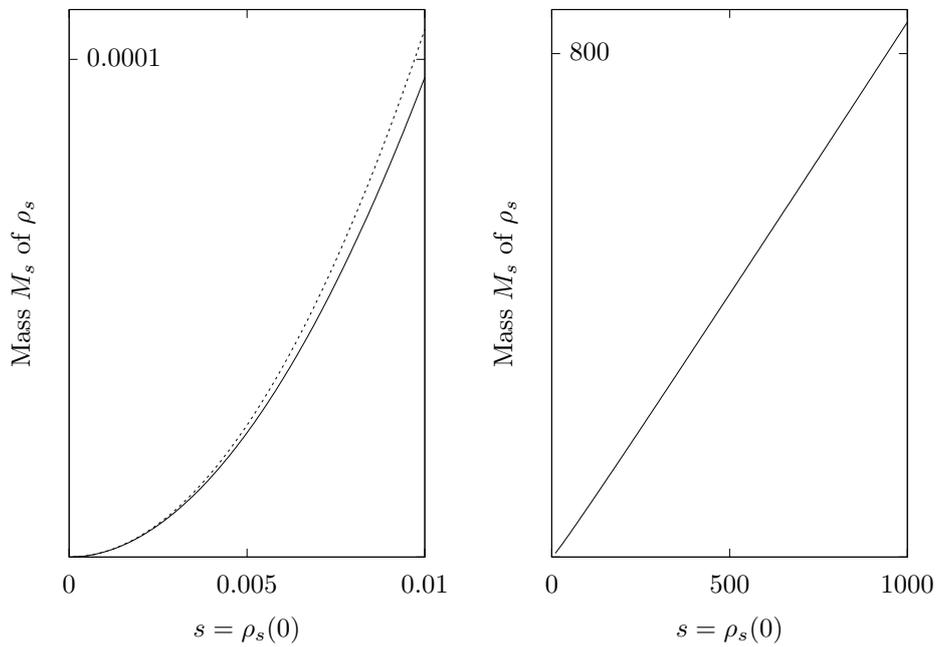

		\begin{center}
			\include{figure-mass-of-minimizer-limits}
		\end{center}
		\caption{The two diagrams show the same information as Figure \ref{figure-mass-of-minimizer} but with different ranges. In the left diagram we zoomed in on the range $s\in[0.0001,0.01]$. This is the deep MOND limit. In this range the graph of $M_s$ is approximately a parabola. For comparison we plotted (dashed line) also the parabola $1.06\,s^2$. In the right diagram we zoomed out on the range $s\in[10,1000]$. This is the Newtonian limit. There the graph of $M_s$ becomes a straight line. Putting the informations from Figure \ref{figure-mass-of-minimizer} and Figure \ref{figure-mass-of-minimizer-limits} together, we conclude that for every given mass $M>0$ there is exactly one $s>0$ such that $\rho_s$ has mass $M$.}
		\label{figure-mass-of-minimizer-limits}
	\end{figure}

	\bibliographystyle{mnras}
	\bibliography{bibliography_math,bibliography_phys,bibliography_mond}

	\appendix
	
	\section{Appendix}

	We give a proof that the Mondian potential energy of a density $\rho$ is finite, provided that $\rho$ has finite support.
	
	\begin{lem} \label{lemma Mondian potential energy is finite}
		Assume that \ref{lambda bounded from above} holds. Let $p>3$ and $\rho,\bar\rho\in L^1\cap L^p(\R^3)$, $\geq 0$ with compact support and $\|\rho\|_1=\|\bar\rho\|_1$, then
		\begin{equation*}
		\int \left| \scriptQ\left( \left| \nabla U^N_\rho \right| \right) - \scriptQ\left( \left| \nabla U^N_{\bar\rho} \right| \right) \right| \diff x < \infty
		\end{equation*}
		and hence
		\begin{align*}
		\Epot^M(\rho) = - \frac{1}{8\pi} \int  \left| \nabla U^N_\rho \right|^{2} \diff x - \frac{1}{4\pi} \int \left( \scriptQ\left( \left| \nabla U^N_\rho \right| \right) - \scriptQ\left( \left| \nabla \bar U^N \right| \right) \right) \diff x
		\end{align*}
		is finite.
	\end{lem}

	For the proof of this lemma we need a technical proposition.
	
	\begin{prop} \label{prop asymptotic behaviour of gradient UN}
		Let $p>1$, $R>0$ and $\rho\in L^1\cap L^p(\R^3),\geq 0$ with $\supp \rho\subset B_R$. Then
		\begin{equation*}
		\sqrt{1-\frac{R^2}{|x|^2}} \frac{\|\rho\|_1}{(|x|+R)^2} \leq \left|\nabla U^N_\rho(x)\right| \leq \frac{\|\rho\|_1}{(|x|-R)^2}
		\end{equation*}
		for a.e. $x\in\R^3$ with $|x|> R$.
	\end{prop}
	
	\begin{proof}
		Let $x,y\in\R^3$ with $|y|<R<|x|$, then $|x-y|\geq |x|-R$, and hence
		\begin{equation*}
		\left|\nabla U^N_\rho(x)\right| \leq \int\frac{\rho(y')}{(|x|-R)^2} \diff y' = \frac{\|\rho\|_1}{(|x|-R)^2}.
		\end{equation*}
		Let $\alpha$ be the angle between the vectors $x$ and $x-y$. With $\alpha_0$ as in the following sketch
		\begin{center}
			\includegraphics{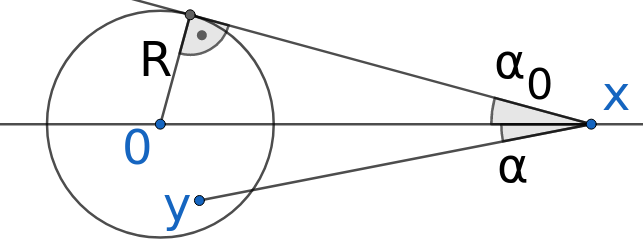}
		\end{center}
		we have
		\begin{equation*}
		|\alpha| \leq \alpha_0 \leq \frac{\pi}{2}
		\end{equation*}
		and
		\begin{equation*}
		\cos \alpha \geq \cos \alpha_0 = \sqrt{1-\sin^2\alpha_0} = \sqrt{1-\frac{R^2}{|x|^2}}.
		\end{equation*}
		Together with the estimate $|x-y|\leq |x|+R$, this yields
		\begin{equation*}
		\left|\nabla U^N_\rho(x)\right| \geq \left| \nabla U^N_\rho(x)\cdot\frac{x}{|x|} \right| = \left| \int \frac{\cos\alpha}{|x-y'|^2}\rho(y') \diff y' \right| \geq \sqrt{1-\frac{R}{|x|^2}} \frac{\|\rho\|_1}{(|x|+R)^2}.
		\end{equation*}
	\end{proof}
	
	\begin{proof}[Proof of Lemma \ref{lemma Mondian potential energy is finite}]
		Lemma \ref{lemma Newtonian potential} implies that $D^2U^N_\rho, D^2U^N_{\bar\rho}\in L^p(\R^3)$. Since $p>3$, Morrey's inequality \cite[Section 5.6. Theorem 5]{2010Evans} implies that there is a $C>0$ such that for every $x\in\R^3$
		\begin{equation*}
		\|\nabla U^N_\rho\|_{L^\infty(B_1(x))},\, \|\nabla U^N_{\bar\rho}\|_{L^\infty(B_1(x))} < C.
		\end{equation*}
		Hence
		\begin{equation*}
		\nabla U^N_\rho, \, \nabla U^N_{\bar\rho} \in L^\infty(\R^3).
		\end{equation*}
		Hence for $R>0$ 
		\begin{align*}
		\int_{|x|<2R} \left| \scriptQ  \left( \left| \nabla U^N_\rho \right| \right) - \scriptQ\left( \left| \nabla U^N_{\bar\rho} \right| \right) \right| \diff x  <  \infty.
		\end{align*}
		Fix $R>0$ such that $\supp\rho,\supp\bar\rho \subset B_R$. Using Lemma \ref{lemma Q continuous} we can estimate
		\begin{align*}
		\int_{|x|\geq 2R} \left| \scriptQ  \left( \left| \nabla U^N_\rho \right| \right) - \scriptQ\left( \left| \nabla U^N_{\bar\rho} \right| \right) \right| \diff x \leq C  \int_{|x|\geq 2R}  \left( \left| \nabla U^N_\rho \right|^{3/2}  -  \left| \nabla U^N_{\bar\rho} \right|^{3/2}  \right).
		\end{align*}
		Using Proposition \ref{prop asymptotic behaviour of gradient UN} we can estimate further
		\begin{align*}
		\int_{|x|\geq 2R} & \left| \scriptQ  \left( \left| \nabla U^N_\rho \right| \right) - \scriptQ\left( \left| \nabla U^N_{\bar\rho} \right| \right) \right| \diff x \\
		\leq & \,C \|\rho\|_1^{3/2} \int_{|x|\geq 2R} \left( \frac{1}{(|x|-R)^3} - \left(1-\frac{R^2}{|x|^2}\right)^{3/4} \frac{1}{(|x|+R)^3}\right) \diff x \\
		\leq &\, C \int_{|x|\geq 2R} \left(\frac{1}{(|x|-R)^3} -\frac{1}{(|x|+R)^3}\right) \diff x \\
		& + C\int_{|x|\geq 2R} \frac{1}{(|x|+R)^3}\left( 1 - \left(1-\frac{R^2}{|x|^2}\right)^{3/4} \right) \diff x \\
		\leq& \,C \int_{|x|\geq 2R} \frac{6R}{(|x|-R)^4}\diff x + C\int_{|x|\geq 2R} \frac{1}{(|x|+R)^3}\frac{R^{3/2}}{|x|^{3/2}} \diff x < \infty.
		\end{align*}
		So
		\begin{equation*}
		\int \left| \scriptQ\left( \left| \nabla U^N_\rho \right| \right) - \scriptQ\left( \left| \nabla U^N_{\bar\rho} \right| \right) \right| \diff x < \infty.
		\end{equation*}
		In particular this implies that the second integral in the difference of the potential energies exists. Further Lemma \ref{lemma Newtonian potential} implies that the first integral exists. So
		\begin{equation*}
		\tilde E_{pot}(\rho) - \tilde E_{pot}(\bar\rho)
		\end{equation*}
		is well defined and finite.
	\end{proof}
	
\end{document}